\documentclass[sigconf, balance=false]{acmart}
\usepackage{popets}

\setcopyright{popets}
\copyrightyear{YYYY}

\acmYear{YYYY}
\acmVolume{YYYY}
\acmNumber{X}
\acmDOI{XXXXXXX.XXXXXXX}
\acmISBN{}
\acmConference{Proceedings on Privacy Enhancing Technologies}
\usepackage{lipsum}
\usepackage{enumitem}
\usepackage{subcaption}

\usepackage[disable]{todonotes}

\newtheorem{theorem}{Theorem}
\newtheorem{lemma}{Lemma}
\newtheorem{corollary}{Corollary}
\theoremstyle{definition}
\newtheorem{definition}{Definition}

\usepackage{tikz}
\usepackage{tkz-euclide}
\usetikzlibrary{arrows.meta, angles, quotes, chains,positioning}

\allowdisplaybreaks

\newcommand{\vtr}[1]{\mathbf{#1}}
\newcommand{\norm}[1]{\lVert{#1}\rVert}
\newcommand{\dotprod}[2]{\langle{#1}, {#2}\rangle}
\newcommand{\dx}{d_\mathcal{X}}
\newcommand{\noise}{\boldsymbol{\eta}}

\begin{document}

\title{\texorpdfstring{$d_\mathcal{X}$}{dX}-Privacy for Text and the Curse of Dimensionality}{\thanks{This is the full version of the paper accepted for publication in the Proceedings on the 26th Privacy Enhancing Technologies Symposium (PoPETs) 2026.}}

\author{Hassan Jameel Asghar}
\email{hassan.asghar@mq.edu.au}
\orcid{0000-0001-6168-6497}
\affiliation{%
	\institution{Macquarie University}
	\city{Sydney}
    \state{}
	\country{Australia}
}
\author{Robin Carpentier}
\email{robin.carpentier@mq.edu.au}
\orcid{0000-0003-1369-2248}
\affiliation{%
	\institution{Macquarie University}
	\city{Sydney}
    \state{}
	\country{Australia}
}
\author{Benjamin Zi Hao Zhao}
\email{ben\_zi.zhao@mq.edu.au}
\orcid{0000-0002-2774-2675}
\affiliation{%
	\institution{Macquarie University}
	\city{Sydney}
    \state{}
	\country{Australia}
}
\author{Dali Kaafar}
\email{dali.kaafar@mq.edu.au}
\orcid{0000-0003-2714-0276}
\affiliation{%
	\institution{Macquarie University}
	\city{Sydney}
    \state{}
	\country{Australia}
}


\begin{abstract}
A widely used method to ensure privacy of unstructured text data is the multidimensional Laplace mechanism for $\dx$-privacy, which is a relaxation of differential privacy for metric spaces. We identify an intriguing peculiarity of this mechanism. When applied on a word-by-word basis, the mechanism either outputs the original word, or completely dissimilar words, and very rarely outputs semantically similar words. We investigate this observation in detail, and tie it to the fact that the distance of the nearest neighbor of a word in any word embedding model (which are high-dimensional) is much larger than the relative difference in distances to any of its two consecutive neighbors. We also show that the dot product of the multidimensional Laplace noise vector with any word embedding plays a crucial role in designating the nearest neighbor. We derive the distribution, moments and tail bounds of this dot product. We further propose a fix as a post-processing step, which satisfactorily removes the above-mentioned issue. 
\end{abstract}

% <concept_id>10002978</concept_id>
% <concept_desc>Security and privacy</concept_desc>
% <concept_id>10002978.10002991.10002995</concept_id>
% <concept_desc>Security and privacy~Privacy-preserving protocols</concept_desc>

% \ccsdesc[500]{Security and privacy}
% \ccsdesc[500]{Security and privacy~Privacy-preserving protocols}



%\renewcommand\footnotetextcopyrightpermission[1]{} % removes footnote with conference information in first column
\settopmatter{printfolios=true, printacmref=false}

\maketitle

\section{Introduction}
Unstructured text is the most common method of communication in the real-world, in the form of emails, sharing reports, and entering prompts to generative artificial intelligence (AI) models to quote a more contemporary example. In many scenarios, part of the text needs to be sanitized due to privacy reasons. Examples include authorship anonymization to protect identities of whistleblowers, redacting sensitive information before releasing documents as part of freedom of information requests, and ensuring that user submitted prompts to a third-party generative AI model do not contain sensitive company information. 

With the increasing use of differential privacy~\cite{dwork2006calibrating} as a principled way of releasing data with privacy in many real-world applications, it has also been proposed as an automated way to achieve privacy in the text domain. More specifically, a generalization of differential privacy to metric spaces, called $\dx$-privacy~\cite{chatzikokolakis2013broadening}, has been proposed as a method to sanitize sensitive text~\cite{fernandes2019generalised, feyisetan2020dxprivacy, qu2021bert, yue2021dxprivacy, li2023llm-prompt, mattern-limits-word-level-dp}.\footnote{$\dx$-privacy is also commonly known as \emph{metric privacy}.} Informally $\dx$-privacy ensures that it is harder to distinguish objects in a metric space that are closer to one another under the distance metric of the metric space than objects further away. This promises better utility than ordinary differential privacy as in many use cases it suffices to provide privacy up to a certain granularity. An analogy is location data; disclosing the city or postcode one resides in is less of a concern than the exact street address. 

There are a number of ways in which $\dx$-privacy can be used to santize text data. One common method is word-by-word sanitization through the \emph{word-level multidimensional Laplace} mechanism for $\dx$-privacy, applied in the following manner~\cite{fernandes2019generalised, feyisetan2020dxprivacy, qu2021bert, yue2021dxprivacy, li2023llm-prompt}. We assume a pre-trained machine learning model is available which vectorizes words, i.e., converts them into embeddings in a high-dimensional space. This pre-trained model is public information. The words form the vocabulary, and their corresponding embeddings define the embedding space. Given a word in a sentence to be sanitized, a noise vector is generated calibrated according to the privacy parameter $\epsilon$, and then added to the word embedding resulting in a vector in the embedding space. Almost always, this does not correspond to the embedding of any word in the vocabulary of the embedding model. We can instead find the nearest neighbor (e.g., with respect to the Euclidean distance) to this noisy embedding in the embedding space, and output the resulting word. This method can be applied independently to each word, and the complete sentence can then be used as the sanitized text.\footnote{Apart from this use case, this mechanism has also been used for author obfuscation~\cite{fernandes2019generalised}, and to provide privacy where each user's input is a set of one or more words~\cite{feyisetan2020dxprivacy}.} 

When sanitizing text through this method we encountered a confounding observation~\cite{robin2024santizellm}. 
If this method is applied several times on the same word, one expects to see the original word, followed by its closest neighbors as the most frequent words output by the mechanism, with the frequency dropping smoothly but exponentially as we move away from the original word. However, we observed that through the entire spectrum of values of $\epsilon$, the mechanism almost always outputs either the original word or words which are very far off in distance and semantic similarity to the original word. The nearest neighbors of the original word are seldom output by the mechanism~\cite{robin2024santizellm}. 

One may hasten to attribute this phenomenon to the high dimensional nature of the embedding space, since the Euclidean distance metric is known to suffer from the so-called \emph{curse of dimensionality}~\cite{beyer1999nn-useful, aggarwal2001surprising-behavior-distance}. However, this is not true in this instance, or not true in the way we may think, since the word embedding models are trained to ensure that the Euclidean distance is an effective metric to identify similar words in the embedding space~\cite{glove, word2vec}. See Section~\ref{sec:rw} for a detailed discussion on this topic. 

%The negative results related to Euclidean metric in higher dimensions in~\cite{beyer1999nn-useful}, for instance, apply to independent, identically distributed dimensions, and not to data with correlated dimensions~\cite{durrant2009nn-useful-converse}. The embedding space contains correlated clusters of words, which explains why the Euclidean metric behaves as expected despite the high-dimensional space.  

In this paper, we investigate this observation in detail and identify the post-processing step of the nearest neighbor search of the perturbed embedding in a high dimensional space as the culprit. Along the way, we unravel a number of other contributions which we believe will prove useful to construct better $\dx$-privacy mechanisms for text data. The multidimensional Laplace mechanism is fundamental to $\dx$-privacy due to its ease of implementation, and hence wide-spread use, akin to the Laplace mechanism~\cite{dwork2006calibrating} for ordinary differential privacy.

\begin{itemize}[leftmargin=*]
    \item We highlight the above-mentioned issue in a very commonly used $\dx$-privacy mechanism for text data~\cite{fernandes2019generalised, feyisetan2020dxprivacy, qu2021bert, yue2021dxprivacy, li2023llm-prompt, mattern-limits-word-level-dp}, whereby across different values of $\epsilon$ and different word embedding models, almost always either the word is not replaced, or is replaced by a completely dissimilar word. 
    \item While analyzing the above phenomenon it turns out that the dot product of the noise vector against any word embedding plays a crucial role. This distribution, which we call the \emph{noisy dot product} distribution, has a length component, which is known to follow the gamma distribution~\cite{feyisetan2020dxprivacy}, and an angular component. We derive the probability density function of the angular component and its moments. 
    \item We show that the angular component is sub-Gaussian with variance $1/n$, where $n$ is the number of dimensions. This means that the cosine of the angle of the noise vector with any embedding is within $\mathcal{O}(1/\sqrt{n})$, and hence the noise vector is increasingly orthogonal to any embedding regardless of the distribution of words in the embedding model. We further prove tail-bounds on the noisy dot product distribution showing that its mass is concentrated within $\mathcal{O}(\sqrt{n}/\epsilon)$.
    \item Through our analysis, we show that the aforementioned observation is related to the fact that in high dimensions the nearest neighbor of a word is more distant than the relative difference of the distances of its two nearest neighbors. We prove necessary conditions on the initial word to be output by the mechanism as opposed to its close neighbors by relating it to the noisy dot product distribution. 
    Previous works~\cite{carvalho2023tem, xu2020mahalanobis, chatzikokolakis2015elastic, biswas2023privic} have highlighted the vulnerability of outliers in $\dx$-privacy. However, their observation is related to the data distribution containing some isolated points. In contrast, our work shows that in high-dimensional settings the outlier issue is the norm rather than the exception.
    \item We propose a possible mitigation as a further post-processing step, and show that the resulting mechanism gives better utility and behaves as expected. The advantage of our proposed fix is that it does not amend the original mechanism, and only adds an extra step. We have released the code of our experiments, including this fix, to promote reproducibility.\footnote{\url{https://github.com/r-carpentier/dx-privacy-curse}\label{footnoteCodeRepo}}
\end{itemize}

\section{The Unusual Behavior of Word-Level \texorpdfstring{$\dx$}{dX}-Privacy}
\label{sec:unusual}
\subsubsection*{The Observation}
We first show an example of applying $\dx$-privacy on text data with the algorithm outlined in the introduction. Formal introduction to $\dx$-privacy and concrete details of this algorithm are presented in Section~\ref{subsec:dx-privacy}. Consider the following text:
%\robin{Example filled with epsilon = 10, Glove 100d}
\begin{quote}
    Maria Gonzalez, a patient at Riverside Clinic, was diagnosed with depression on March 5, 2023. She currently lives at 789 Oak Drive, San Francisco. Maria has been prescribed medication and is undergoing weekly therapy sessions.
\end{quote}
To sanitize this text we first choose an embedding model. Let us say we use the \verb+GloVe+ embedding model~\cite{glove} with $n = 100$ dimensions. We then take the first word of the text, pass it through the embedding model to obtain its embedding $\vtr{w}$. We then sample a noise vector $\noise$ according to a distribution scaled to $\epsilon$ and $n$. For this example, we choose $\epsilon = 10$. We thus obtain the noisy embedding $\vtr{w}^* = \vtr{w} + \noise$. This does not correspond to any word in the vocabulary of \verb+GloVe+. We therefore, search the nearest neighbor in the embedding space of \verb+GloVe+ to $\vtr{w}^*$ and output the corresponding word as the replacement to the original word. This process is repeated for each word\footnote{Commas and periods in the original text were excluded from the sanitization for readability} resulting in the following sanitized text: 
\begin{quote}
    maria carvalho, full patient raised bottomland clinic, was diagnosed with evanston on 8 4, 2028 . she represent lives ` 789 poplar drive, st. antonio. maria deeply were prescribed medication deadlock subject undergone quiz therapy approaches.
\end{quote}
In most cases, as above, the sentence is ill-structured and grammatically incorrect. This can be corrected by another post-processing step, for example by leveraging a generative AI model. The following text is obtained by asking \verb+ChatGPT 4o+\footnote{See \url{https://chatgpt.com}.} to correct the grammar of the previous text:
\begin{quote}
    Maria Carvalho, a full-time patient at Bottomland Clinic, was diagnosed in Evanston on August 4, 2028. She resides at 789 Poplar Drive, St. Antonio. Maria was prescribed medication and has undergone various therapeutic approaches, including cognitive therapy.
\end{quote}
As we can see after this step, the sentence is coherent with some of the exact names and dates replaced, which is desirable for privacy. 
Now, if we use smaller values of $\epsilon$ (more privacy) we expect the original words to be replaced by their distant neighbors, i.e., words that have little to no similarity with the original word. However, as we increase $\epsilon$ (less privacy), we would expect most words to either remain unchanged or replaced by synonyms. At least this is what we expect if we apply the usual Laplace mechanism~\cite{dwork2006calibrating} of differential privacy to one-dimensional data. Figure~\ref{fig:lap-1d} shows an example of one-dimensional data, in which we plot the result of applying the Laplace mechanism with different values of $\epsilon$ (assuming sensitivity one) with the universe of values confined to the set of positive integers. The true value is $a = 400,000$, and after adding noise we round it to the nearest integer. Any integer value within $a \pm 100$ is considered a close neighbor of $a$ and all other values considered distant neighbors. For each value of $\epsilon$ we sample $10,000$ noisy answers and report the proportion of times the original, close neighbors and distant neighbors are output by the mechanism. Initially with $\epsilon = 0.001$, distant neighbors are output more frequently (Figure~\ref{fig:lap-1d}, left). This trend is quickly flipped as we increase $\epsilon$. With $\epsilon > 1$ we see that it is either the original value or the close neighbors that are output by the mechanism and very rarely any distant neighbors (Figure~\ref{fig:lap-1d}, right). 

% \begin{figure}[!ht]
%     \centering
%     \begin{subfigure}[t]{0.49\linewidth}
%         \includegraphics[width=\linewidth]{figs/lap-1d-close-up.pdf}
%         \caption{Close-up perspective of $\epsilon$}
%         \label{subfig:lap-1d-close-up}
%     \end{subfigure}
%     % \hspace{10mm}%
%     \begin{subfigure}[t]{0.49\linewidth}
%         \includegraphics[width=\linewidth]{figs/lap-1d.pdf}
%         \caption{Greater view of $\epsilon$}
%         \label{subfig:lap-1d}
%     \end{subfigure}
%     \caption{The proportion of times the original value, its close neighbors and distant neighbors are output by the Laplace mechanism of differential privacy. The original value is $a = 400,000$, close neighbors are all integers within $a \pm 100$, and all other integer values are distant neighbors.}
%     \label{fig:lap-1d}
% \end{figure}
\begin{figure}[!ht]
    \centering
    \includegraphics[width=\columnwidth]{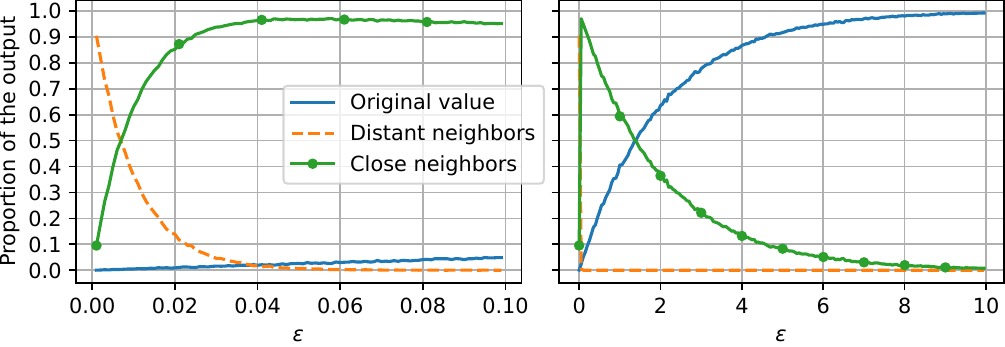}
    \caption{The proportion of times the original value, its close neighbors and distant neighbors are output by the Laplace mechanism of differential privacy. The original value is $a = 400,000$, close neighbors are all integers within $a \pm 100$, and all other integer values are distant neighbors.}
    \label{fig:lap-1d}
\end{figure}

However, we observed that through the text counterpart of this algorithm, i.e., the multidimensional Laplace mechanism, either the initial word is output by the mechanism, or very far off words, and the nearest neighbors are seldom encountered. This is visualized in Figure~\ref{fig:lap-md} where we show the proportion of times the original word is output by the mechanism, against its close neighbors, the first 100 nearest neighbors, and distant neighbors, which constitute the rest of the words in the vocabulary. These plots were obtained by randomly sampling 5,000 words. 

At lower values of $\epsilon$ only far off words are selected, but as we increase $\epsilon$, the original word dominates at the expense of its close neighbors. The end-result of course is bad privacy-utility tradeoff. The figure shows the pattern for the \verb+GloVe-Wiki+ and \verb+fastText+ embedding models, both with 300 dimensions. The result is similar for other embedding models which are detailed in Section~\ref{subsec:word-embeddings}. 
For a more detailed description of this observation please see~\cite{robin2024santizellm}.

These results are also consistent with the results shown in~\cite{feyisetan2020dxprivacy, qu2021bert}, where the authors show the rather high frequency of unmodified words but no commentary is provided as to why this may be the case other than attributing it to the behavior of the mechanism for higher values of $\epsilon$. Indeed, the trend is more peculiar at smaller values of $\epsilon$. 

\begin{figure}[!ht]
    \centering
    \includegraphics[width=\columnwidth]{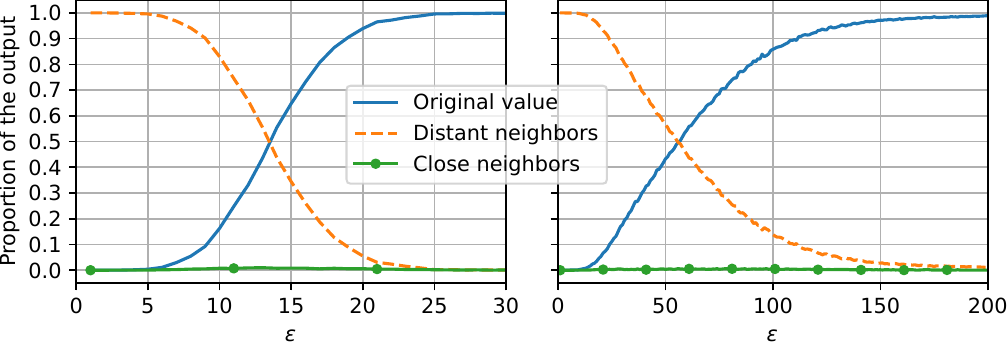}
    \caption{The proportion of times the original word, its close neighbors and distant neighbors are output by the multidimensional Laplace mechanism of $\dx$-privacy on the \texttt{Glove-Wiki} (left) and \texttt{Word2Vec} (right) word embedding models. Close neighbors are the first 100 nearest neighbors, and all other words are distant neighbors.}
    \label{fig:lap-md}
\end{figure}

\subsubsection*{An Illustrative Two-Dimensional Example}
\label{subsec:illustration}
The issues with $\dx$-privacy over high-dimensional word embeddings can be illustrated through an analogy of $\dx$-privacy for location data.\footnote{$\dx$-privacy for location data is studied in detail in~\cite{andres2013geo} where the authors term the notion, geo-indistinguishability.} Assume that we have a dataset containing the resident states of the inhabitants of the United States of America (USA). We exclude the state of Alaska in this example. We wish to release this dataset with privacy, meaning that an individual can plausibly deny that he/she resides in a particular state in the released dataset. To do so, we perturb each individual's location using $\dx$-privacy. Assume each person's location is given as a point on the real plane (latitude and longitude, if you like). One way to satisfy $\dx$-privacy is to sample a noise vector, add it to the original location, and then find the nearest state to the perturbed location. The last step is a post-processing step, and maintains $\dx$-privacy. The noise is a two-dimensional vector sampled from a particular distribution~\cite{andres2013geo} scaled by the privacy parameter $\epsilon$ and the distance metric is the Euclidean distance in practice which we describe in detail in Section~\ref{sec:bg}. This is exactly the mechanism commonly used to provide (word-level) $\dx$-privacy in text data~\cite{fernandes2019generalised, feyisetan2020dxprivacy, qu2021bert, yue2021dxprivacy, li2023llm-prompt}, except that location coordinates are replaced by word embeddings and states are replaced by words in the vocabulary.

Now consider a resident of the state of Hawaii. For values of $\epsilon$ within a certain range, the sampled perturbed location will fall in the North Pacific Ocean with overwhelming probability. Since our universe of locations is the set of states in USA, the nearest location is Hawaii again. Compare this to a resident of Kansas. For the same value of $\epsilon$, the noisy location in this case is likely to land on the neighboring states and beyond. Figure~\ref{fig:usa-map-illustrate} illustrates this point. This observation has both privacy and utility implications. 

\begin{figure}[h]
    \centering
    \includegraphics[width=1\linewidth]{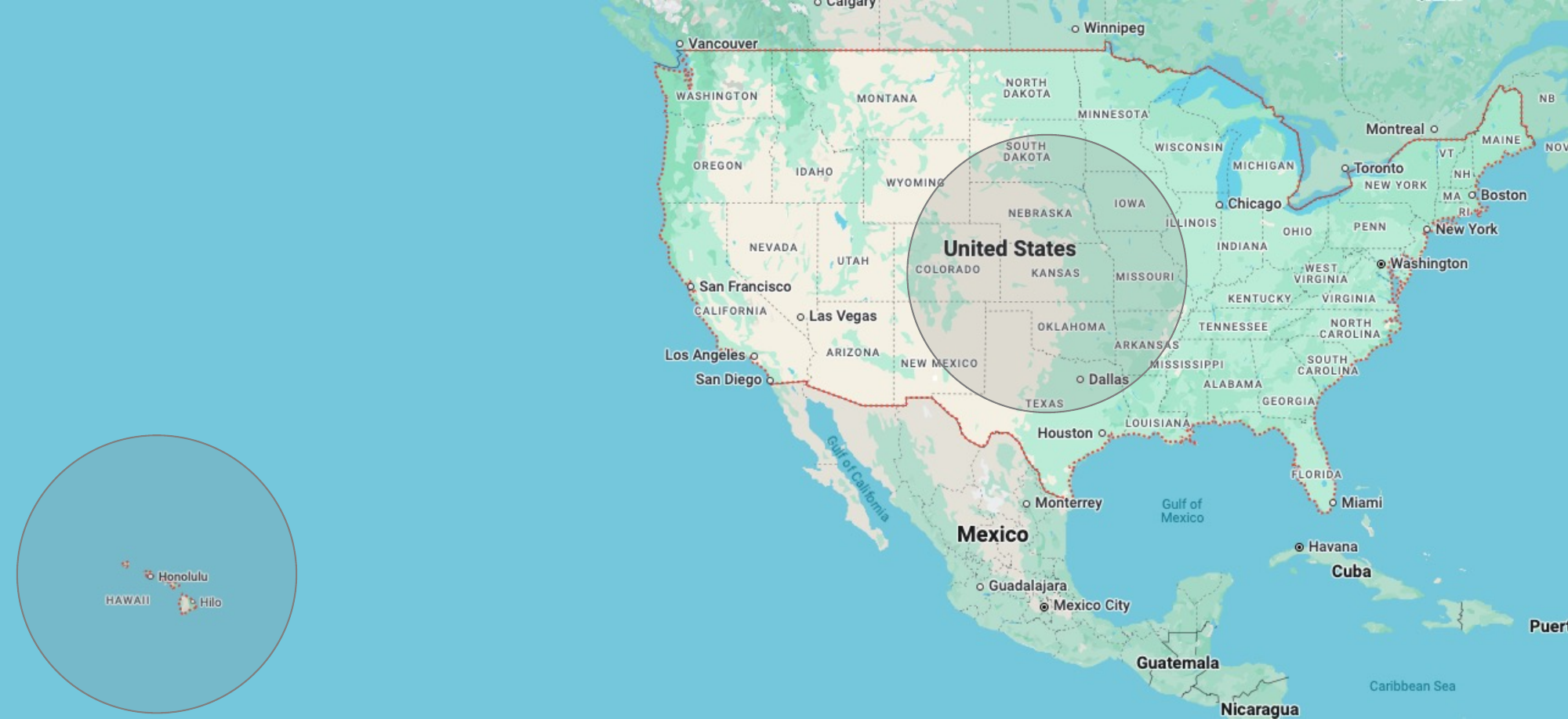}
    \caption{An analogy of the issues of applying $\dx$-privacy to high-dimensional word embeddings using location data. The nearest neighbor of the perturbed location of the residents of Hawaii tends to be Hawaii itself, unlike the states in mainland USA. The circles indicate the probability density of the noise vector. The map is extracted from Google Maps.}
    \label{fig:usa-map-illustrate}
\end{figure}

From the privacy point-of-view, the residents of mainland USA get more protection than the residents of Hawaii for the same value of $\epsilon$. In other words, the level of privacy provided depends on the structure of the dataset. Datasets with with irregularly scattered points will end up providing less privacy for isolated points. This issue has already been highlighted in previous works~\cite{carvalho2023tem, xu2020mahalanobis}. Here, we would like to point out that this privacy issue does not arise if we scale the noise according to the sensitivity of the distance function which makes the resulting privacy notion equivalent to ordinary $\epsilon$-differential privacy~\cite{dwork2006calibrating}. However, applications of this $\dx$-privacy mechanism to text domain do not use sensitivity to scale noise~\cite{fernandes2019generalised, feyisetan2020dxprivacy, qu2021bert, yue2021dxprivacy, li2023llm-prompt}; the idea being that we only need to make it hard to distinguish between nearby points. For location based applications, one specifies a radius, which serves as a proxy for sensitivity, and we say that the corresponding mechanism provides $\dx$-privacy for locations within this radius~\cite{andres2013geo, kamalaruban2020not}. 

One may also dismiss the above observation as an outlier; an issue that relates only to a few isolated points in the dataset. However, as we shall demonstrate, in higher dimensions the nearest neighbor of any word is at a considerable distance away from it, and the distance increases as we increase the size of dimension. Thus, in higher dimensions, this is more the norm than an anomaly.

The above observation also creates issues from a utility perspective. With higher values of $\epsilon$ (less privacy), the original word is returned by the nearest neighbor search most of the times. As we decrease $\epsilon$ to provide more privacy, one expects the frequency of the nearest neighbors of the original word to be selected more than distant words. However, since the variance of the noise is now larger, this happens less often than expected. The undesired result is that either the word is not changed at all, or if it is changed, it is replaced by a distant word with little to no semantic similarity with the original word. Again this issue is exacerbated in higher dimensions, as the distance between a word and its nearest neighbor is higher than the relative difference in the distances to its two successive neighbors.

\section{Background and Notation}
\label{sec:bg}
\subsubsection*{Notations} The $n$-dimensional real space is denoted by $\mathbb{R}^n$. A vector from $\mathbb{R}^n$ will be denoted in bold face, e.g., $\vtr{x}$. Let $\norm{\vtr{x}} = \sqrt{\sum_{i=1}^n |x^2_i|}$ denote the Euclidean norm of $\vtr{x}$, where $x_i$ is the $i$th element of $\vtr{x}$. The dot product between two vectors $\vtr{x}$ and $\vtr{y}$ is denoted as $\dotprod{\vtr{x}}{\vtr{y}}$. The following is an elementary fact:
\[
\dotprod{\vtr{x}}{\vtr{y}} = \norm{\vtr{x}} \norm{\vtr{y}} \cos \theta_{\vtr{x}, \vtr{y}},
\]
where $0 \leq \theta_{\vtr{x}, \vtr{y}} \leq \pi$, is the angle between $\vtr{x}$ and $\vtr{y}$.\footnote{Keeping $\vtr{x}$ fixed, if the direction of $\vtr{y}$ is chosen randomly in the plane containing the two vectors, then $\theta_{\vtr{x}, \vtr{y}}$ is the smallest of the two possible angles. This is well-defined. To see this let $\theta_1$ and $\theta_2$ denote the two possible angles, and let $\theta_1 > \pi$. Then $\theta_2 = 2\pi - \theta_1 < \pi$, and note that $\cos \theta_1 = \cos (2\pi - \theta_2) = \cos \theta_2$.} We can interpret $\norm{\vtr{y}} \cos \theta_{\vtr{x}, \vtr{y}}$ as the length of the projection of $\vtr{y}$ on $\vtr{x}$. For any two vectors $\vtr{x}, \vtr{y} \in \mathbb{R}^n$ the Euclidean distance between $\vtr{x}$ and $\vtr{y}$ is $\norm{\vtr{x} - \vtr{y}}$. The Euclidean distance is a metric as it satisfies the following properties of (1) (positivity) $\norm{\vtr{x} - \vtr{y}} \ge 0$ with equality if and only if $\vtr{x} = \vtr{y}$, (2) (symmetry) $\norm{\vtr{x} - \vtr{y}} = \norm{\vtr{y} - \vtr{z}}$, and (3) (triangle inequality) $\norm{\vtr{x} - \vtr{y}} \leq \norm{\vtr{x} - \vtr{z}} + \norm{\vtr{z} - \vtr{y}}$, for any $\vtr{x}, \vtr{y}, \vtr{z} \in \mathbb{R}^n$~\cite{searcoid-metric}.

\subsection{Vector Representation of Words}
\label{subsec:word-embeddings}
In recent years, a number of machine learning models have sprung up which produce vector representations of words, which we call word embeddings for short. These pre-trained word embeddings can be downloaded and used for natural language processing tasks. Notable examples include \verb+Word2Vec+~\cite{word2vec}, \verb+GloVe+~\cite{glove}, and \verb+fastText+~\cite{fasttext}. 
%\footnote{See \url{https://code.google.com/archive/p/word2vec/}}
These word embeddings lie on a high-dimensional real vector space with dimensions ranging from 50 to 300, and even higher. Some embeddings, such as \verb+GloVe+, come in different dimensions. Despite the high-dimensional space the word embeddings maintain distance-based semantic similarities. In other words, the Euclidean distance is an effective method to obtain nearest neighbors of a word in terms of semantic similarity to the target word.\footnote{See \url{https://nlp.stanford.edu/projects/glove/}.} Furthermore, these embeddings maintain linear relationships between words, e.g., king minus queen equals man minus woman~\cite{glove}. For a more detailed introduction to vector representation of words and models producing word embeddings, see~\cite[\S 15]{dive-deep-book}. The pre-trained word embeddings used in this paper are shown in Table~\ref{tab:vocabularies}. We will use the terms `embedding model' and `vocabulary' interchangeably. 

\begin{table}[h!]
\centering
% \resizebox{.9\line}{!}{
 \begin{tabular}{c|c|c} 
 \hline
Vocabulary & Dimensions & Words \\
 \hline\hline
\verb+GloVe-Twitter+ & 25, 50, 100, 200 & 1,193,514 \\
\verb+GloVe-Wiki+ & 50, 100, 200, 300 & 400,000\\
\verb+Word2Vec+ & 300 & 3,000,000  \\
\verb+fastText+ & 300 & 2,519,370 \\ \hline
\end{tabular}
% }
\caption{Word embedding models used in this work.}
\label{tab:vocabularies}
\end{table}

\subsection{\texorpdfstring{$\dx$}{dX}-Privacy and Applications to Text Data}
\label{subsec:dx-privacy}
Let $\mathcal{D} \subseteq \mathbb{R}^n$ be the input domain, which for our purpose is the embedding space. Let $\vtr{x} \in \mathcal{D}$ denote a word. 

\begin{definition}[Differential Privacy~\cite{dwork2006calibrating}]
\label{def:dp}
An algorithm $\mathcal{M}: \mathcal{D} \rightarrow R$ satisfies $\epsilon$-(local) differential privacy if for all words $\vtr{x}, \vtr{y} \in \mathcal{D}$ and for all possible subsets $S \subseteq R$ we have 
\[
\Pr[ \mathcal{M}(\vtr{x}) \in S] \leq e^{\epsilon} \Pr[ \mathcal{M}(\vtr{y}) \in S] 
\]
\end{definition}

\begin{definition}[$\dx$-privacy~\cite{chatzikokolakis2013broadening}]
Let $d$ be a metric on $\mathcal{D}$. An algorithm $\mathcal{M}: \mathcal{D} \rightarrow R$ satisfies $\epsilon \dx$-privacy for the metric $d$ if for all words $\vtr{x}, \vtr{y} \in \mathcal{D}$ and for all possible subsets $S \subseteq R$ we have 
\[
\Pr[ \mathcal{M}(\vtr{x}) \in S] \leq e^{\epsilon d(\vtr{x}, \vtr{y})} \Pr[ \mathcal{M}(\vtr{y}) \in S] 
\]
\end{definition}

Note that if $d$ is the Hamming distance $d_{\text{H}}$ then we recover the original definition of $\epsilon$-differential privacy, since $d_{\text{H}}(\vtr{x}, \vtr{y}) = 1$, whenever $\vtr{x} \neq \vtr{y}$. A key advantage of $\dx$-privacy over ordinary differential privacy is in terms of utility: the former treats all inputs equally, and indistinguishability is with respect to all inputs, whereas the latter provides more privacy with respect to similar inputs than far away inputs, when measured according to the distance metric. Like $\epsilon$-DP, $\dx$-privacy enjoys the properties of immunity to post-processing, and composition of privacy guarantees~\cite{chatzikokolakis2013broadening}. In particular, applying a $\dx$-privacy mechanism independently to each of the words in a sequence (sentence) of $m$ words, makes the resulting composition of these mechanisms $m\epsilon \dx$-private~\cite{kamalaruban2020not}. 

Finally, we would like to point out a particular (simplified) result from the original paper by Dwork et al~\cite{dwork2006calibrating} on $\epsilon$-DP, which relates to general metric spaces. Let $d$ be a metric on the domain $\mathcal{D}$, and let the sensitivity of the distance metric $d$ be defined as
$
\Delta d = \max_{\vtr{x}, \vtr{y} \in \mathcal{D}} d(\vtr{x}, \vtr{y})$,
then the mechanism $\mathcal{M}$ which on input $\vtr{x} \in \mathcal{D}$ outputs a $\vtr{y} \in \mathcal{D}$ with probability proportional to 
\[
\Pr[\mathcal{M}(\vtr{x}) = \vtr{y}] \propto \exp\left( - \frac{\epsilon d(\vtr{x}, \vtr{y})}{2 \Delta d}\right),
\]
is $\epsilon$-differentially private, provided such a probability density function exists~\cite[\S 3.3]{dwork2006calibrating}. We remark that the result in~\cite{dwork2006calibrating} does not contain the negative sign in the proportionality above, but it is easy to verify that the result still holds. As we shall see next, the $\dx$-private mechanism samples a word embedding proportional to above except that there is no scaling according to sensitivity. 

\subsubsection*{$\dx$-Private Noise} In order to add multidimensional Laplace noise to a word embedding $\vtr{w} \in \mathbb{R}^n$, the method is to add a noise vector $\noise \in \mathbb{R}^n$ from a distribution with probability density $ \propto \exp(-\epsilon \norm{\noise})$~\cite{feyisetan2020dxprivacy, dwork2006calibrating}.\footnote{As mentioned above, the result from~\cite{dwork2006calibrating}  is for ordinary $\epsilon$-differential privacy as opposed to $\dx$-privacy.} To sample from this distribution, one first samples $n$ zero-mean, unit-variance Gaussians to produce an $n$-dimentional vector $\vtr{u}$ whose resulting probability density is 
\begin{equation}
\label{eq:unit-gaussian}
\frac{1}{2\pi^{n/2}} e^{(-\frac{\norm{\vtr{u}}}{2})}    
\end{equation}
The vector is then normalized to produce the unit vector $\vtr{\hat{u}} = \frac{\vtr{u}}{\norm{\vtr{u}}}$. This means that $\vtr{\hat{u}}$ is distributed uniformly at random on the surface of a hypersphere of $n$-dimensions with unit radius (see for example~\cite[\S 2.5]{blum2020foundations-ds}). Next we find the ``length'' of the noise vector. The number of points on an $n$-dimensional hypersphere with radius $r$ is proportional to $r^{n-1}$, with each point having density proportional to $\exp(-\epsilon \norm{\noise})$ (the requirement above). 
Thus, we need a probability density function proportional to $r^{n-1} \exp(-\epsilon \norm{\noise})$. This means that we should sample the noise as $\noise = r \vtr{\hat{u}}$, where $r$ has the Gamma distribution~\cite{feyisetan2020dxprivacy}, with probability density
\begin{equation}
\label{eq:gamma-dist}
f_\text{G}(r) = \frac{1}{\Gamma(n)\epsilon^{-n}} r^{n-1}e^{-\epsilon r}.
\end{equation}
Note that $\norm{\noise} = \norm{r \vtr{\hat{u}}} = r$, and hence this is the required distribution. For completeness, we provide a proof in Appendix~\ref{app:sec:gamma}. Figure~\ref{fig:noise-r-2} illustrates the distribution for $n = 2$.  

\begin{figure}[!ht]
    \centering
    \begin{subfigure}[t]{0.18\textwidth}
        \includegraphics[width=\linewidth]{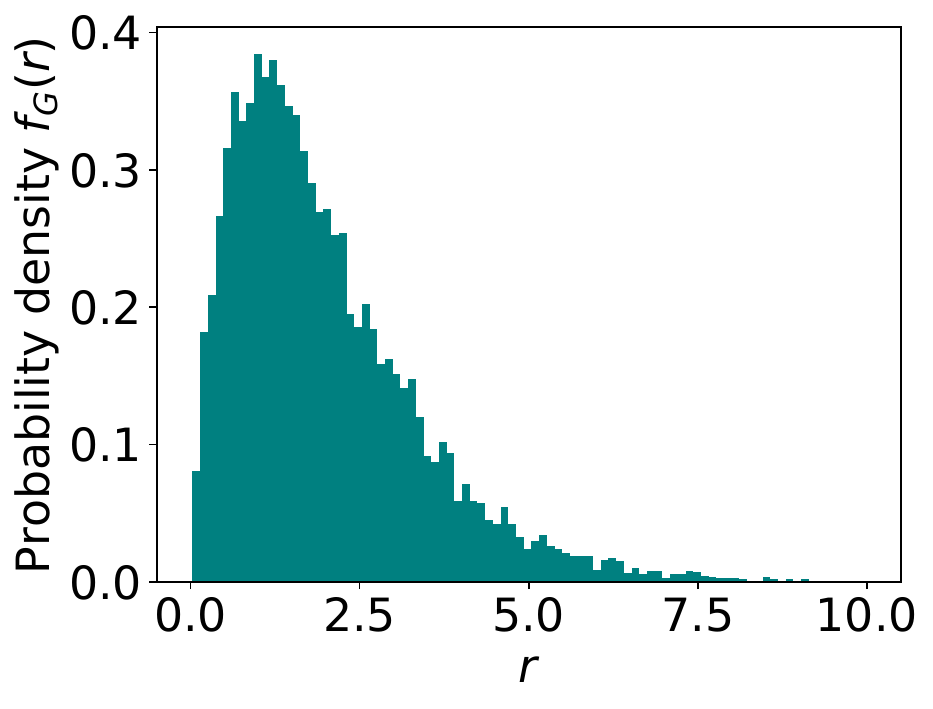}
        \caption{}
        \label{subfig:gamma-2d}
    \end{subfigure}
%    \qquad
    \begin{subfigure}[t]{0.14\textwidth}
        \includegraphics[width=\linewidth]{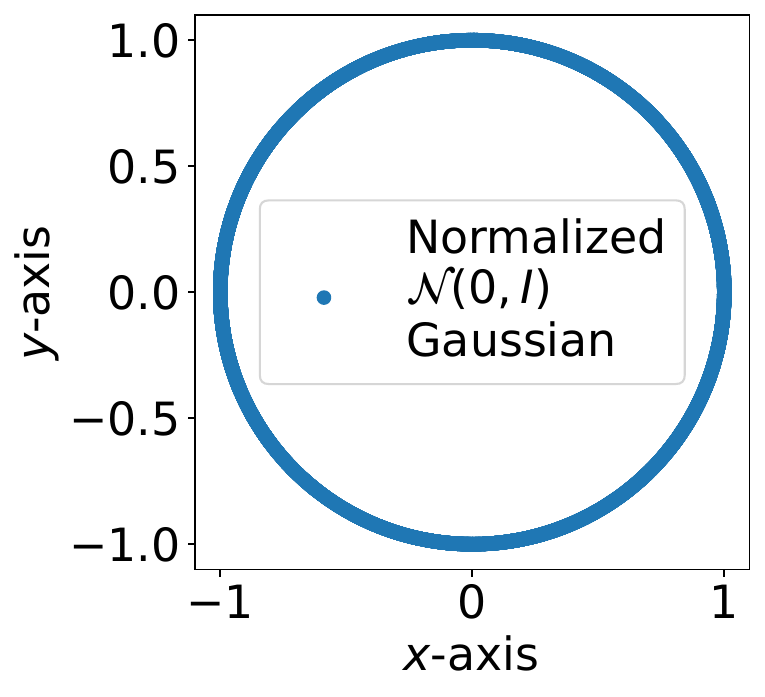}
        \caption{}
        \label{subfig:units-2d}
    \end{subfigure}
%    \qquad
    \begin{subfigure}[t]{0.14\textwidth}
        \includegraphics[width=\linewidth]{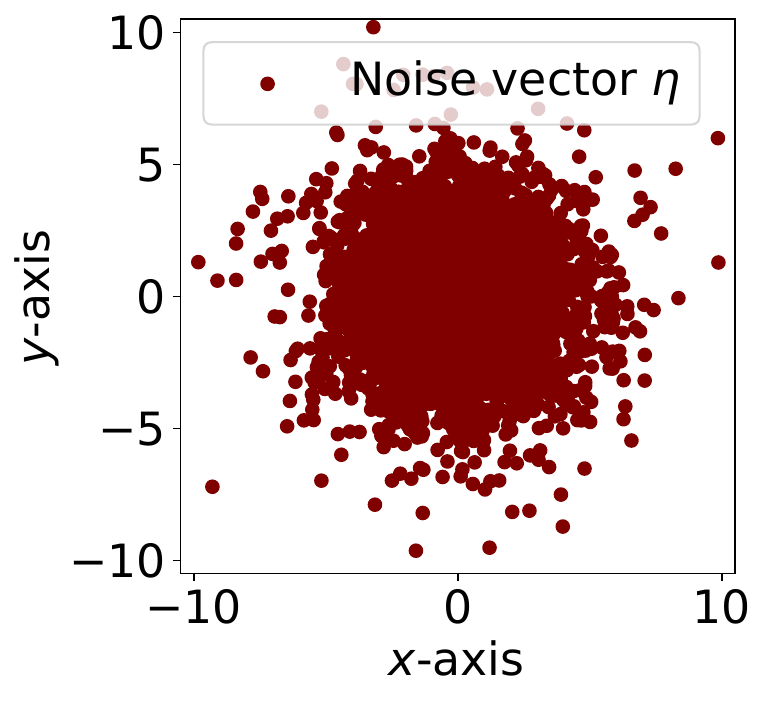}
        \caption{}
        \label{subfig:noise-2d}
    \end{subfigure}
    \caption{The noise distribution for $n=2$ with 10,000 sampled points. Subfigure~\subref{subfig:gamma-2d} shows the distribution of the length of noise vector, i.e., $r$ through Eq.~\eqref{eq:gamma-dist}. Subfigure~\subref{subfig:units-2d} is the unit vector produced by normalizing the 2D Gaussian $\mathcal{N}(\mathbf{0}, I)$ (Eq.~\eqref{eq:unit-gaussian}), where $I$ is the identity matrix. Subfigure~\subref{subfig:noise-2d} shows the distribution of the resulting noise vector $\noise$.}
    \label{fig:noise-r-2}
\end{figure}

\subsubsection*{Nearest Neighbor Search} Almost always, the perturbed embedding $\vtr{w}^* = \vtr{w} + \noise$ is not a member of $\mathcal{D}$. Thus, a nearest neighbor search is performed to find an embedding $\vtr{x}^* \in \mathcal{D}$ which is closest to $\vtr{w}^*$ in the Euclidean distance~\cite{feyisetan2020dxprivacy}. That is, we find the embedding: 
\begin{equation}
\label{eq:min-problem}
    \mathbf{x}^* = \arg \min_{\vtr{x} \in \mathcal{D}} \norm{\vtr{w}^* - \vtr{x}}.
\end{equation}
This vector is then the output of the $\dx$-private mechanism. 
%The hope is that this vector is semantically similar to the original vector $\vtr{x}$.  

\subsubsection*{Proof of Privacy} The above mechanism is $\dx$-private since the resulting noise word embedding $\vtr{w}^*$ is output from the distribution proportional to $\exp(-\epsilon \norm{\noise}) = \exp(-\epsilon d(\vtr{w}^*, \vtr{w}))$, where $d$ is the Euclidean distance. The nearest neighbor search is a post-processing step, and the result can be extended to a sequence (sentence) of multiple words by applying the mechanism independently on each word, and then combining the results, invoking the composition property of $\dx$-privacy. Feyisetan et al also provide a direct proof of this result in~\cite{feyisetan2020dxprivacy}. We note that, while most commonly employed, this is not the only method of applying differential privacy for text sanitization. Other methods include other word-level $\dx$-privacy mechanisms~\cite{feyisetan2019hyperbolic, carvalho2023tem}, sentence-level differential privacy mechanisms~\cite{mattern-limits-word-level-dp, igamberdiev2023dp-bart}, and mechanisms for more advanced natural language processing tasks~\cite{utpala2023document-level}. We describe them in more detail in Section~\ref{sec:rw}. 

\section{Conditions for Nearest Neighbor Selection and the Noisy Dot Product Distribution} 
\label{sec:char-noise}
We are first interested in the conditions when the perturbed embedding $\vtr{w}^*$ is closer to $\vtr{w}$, i.e., the original word, than any of the neighbors of $\vtr{w}$. When this is the case, the original word will be chosen as output by the mechanism. Formalizing these conditions is necessary to understand why the original word is overwhelmingly chosen over its close neighbors (Figure~\ref{fig:lap-md}). To this end, we have the following result.

\begin{theorem}
\label{theorem:pert-emb}
Let $\noise$ be a noise vector. Let $\vtr{w} \in \mathcal{D}$ be a word embedding, and let $\vtr{w}^* = \vtr{w} + \noise$ be the perturbed embedding. Let $\vtr{x} \in \mathcal{D}$ be any embedding different from $\vtr{w}$. Then $\vtr{w}^*$ is closer to $\vtr{w}$ than any of its neighbors if for all neighbors $\vtr{x}$ of $\vtr{w}$, we have 
\[
r \cos \theta_{\noise, \vtr{x} - \vtr{w}} < \frac{1}{2}\norm{\vtr{w} - \vtr{x}}
\]
\end{theorem}
\begin{proof}
Let $\vtr{x} \in \mathcal{D}$ be any neighbor of $\vtr{w}$. We have:
\begin{align*}
    \norm{\vtr{w}^* - \vtr{w}}^2 &< \norm{\vtr{w}^* - \vtr{x}}^2\\
\Rightarrow    \norm{\noise}^2 &< \norm{\vtr{w} - \vtr{x} + \noise}^2\\
\Rightarrow   \norm{\noise}^2 &< \norm{\vtr{w} - \vtr{x}}^2 + \norm{\noise}^2 + 2\dotprod{\vtr{w} - \vtr{x}}{\noise}\\
\Rightarrow -2\dotprod{\vtr{w} - \vtr{x}}{\noise} &< \norm{\vtr{w} - \vtr{x}}^2 \\
\Rightarrow \dotprod{\noise}{\vtr{x} - \vtr{w}} &< \frac{1}{2}\norm{\vtr{w} - \vtr{x}}^2 \\
\Rightarrow r \norm{\vtr{x} - \vtr{w}}\cos \theta_{\noise, \vtr{x} - \vtr{w}} &< \frac{1}{2}\norm{\vtr{w} - \vtr{x}}^2 \\
\Rightarrow r \cos \theta_{\noise, \vtr{x} - \vtr{w}} &< \frac{1}{2}\norm{\vtr{w} - \vtr{x}},
\end{align*}
where in the second last step, we have used Theorem~\ref{theorem:dot-prod}. 
\end{proof}
An interpretation of the above is as follows: The noisy embedding will be closer to the original vector $\vtr{w}$ than its nearest neigbour $\vtr{x}$ if the length of the projection of the noise vector on $\vtr{x} - \vtr{w}$ is less than half the length of $\vtr{x} - \vtr{w}$. Figure~\ref{subfig:orig-vector} illustrates Theorem~\ref{theorem:pert-emb}. This can also be explained easily via a nice geometric illustration. Suppose $0 \leq \theta_{\noise, \vtr{x} - \vtr{w}} \leq \frac{\pi}{2}$, as shown in Figure~\ref{subfig:geo-illustrate}, where we denote the vectors by their end-points in the real space. It is easy to see that if $A < B$, then $C^2 = A^2 - E^2 < B^2 - E^2 = D^2$, and when $C < D$ then $A^2 = C^2 - E^2 < D^2 - E^2 = B^2$. The line segment $C$ is precisely the length of the projection of the noise vector on $\vtr{x} - \vtr{w}$. If $\frac{\pi}{2} \leq \theta_{\noise, \vtr{x} - \vtr{w}} \leq \pi$ then clearly $A < B$, which is supported by the fact that the projection of the noise vector on $\vtr{x} - \vtr{w}$ is negative. 

%Finally if $\pi \leq \theta_{\noise, \vtr{x} - \vtr{w}} \leq 2\pi$, then using the fact that $\cos \theta_{\noise, \vtr{x} - \vtr{w}} = \cos (\pi + \theta_{\noise, \vtr{w} - \vtr{x}}) = -\cos \theta_{\noise, \vtr{w} - \vtr{x}}$, where $0 \leq \theta_{\noise, \vtr{w} - \vtr{x}} \leq \pi$, the same argument applies but this time with the projection of the noise vector on $\vtr{w} - \vtr{x}$.

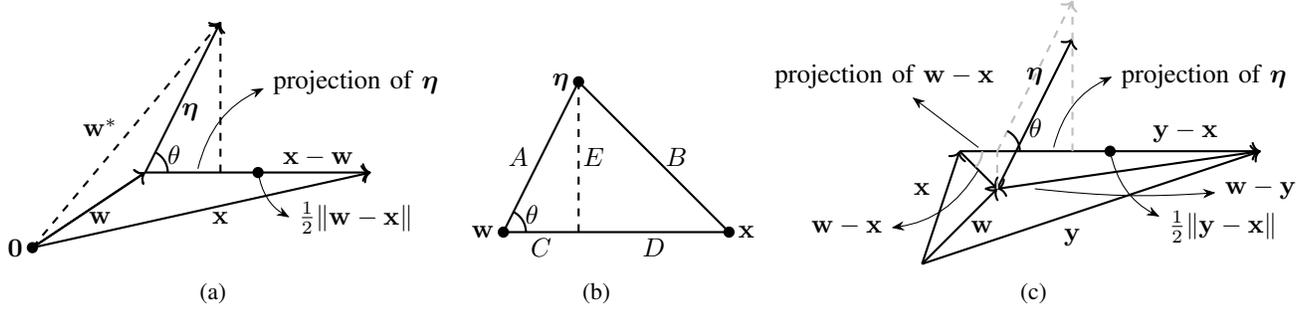
\begin{figure*}[!ht]
\centering
%%%%%%%%%%%%%%%%%%%%%%%%%%%%%
%%%%%%% First figure %%%%%%%%
%%%%%%%%%%%%%%%%%%%%%%%%%%%%%
\begin{subfigure}[t]{0.32\textwidth}
\centering
\begin{tikzpicture}
\draw[thick,->] (0,0) -- (3,0) coordinate[](xwline); % horizontal
\draw[thick, dashed, -] (1,0) -- (1,2);
\draw[thick,->] (0,0) -- (1,2) coordinate[](etaline);
\draw[thick,->] (-1.5,-1) -- (0,0);
\draw[thick,->] (-1.5,-1) -- (0,0);
\draw[thick, dashed, ->] (-1.5,-1) -- (1,2);
\draw[thick, ->] (-1.5,-1) -- (3,0);

\tkzDefPoint(0,0){origin};

\tkzDefPoint(1.5,0){xw};
\node at (xw)[circle,fill,inner sep=1.5pt]{};

\tkzDefPoint(-1.5,-1){o};
\tkzLabelPoint[left](o){$\vtr{0}$};
\node at (o)[circle,fill,inner sep=1.5pt]{};

\node at (2.3,0.2) (xw-vec){${\vtr{x} - \vtr{w}}$};

\node at (-0.6,0.6) (wstar){$\vtr{w}^*$};

\node at (-0.6,-0.6) (w){$\vtr{w}$};

\node at (0.6,0.8) (eta){$\noise$};

\node at (1,-0.6) (x){$\vtr{x}$};

\node at (2.8,1.2) (proj) {projection of $\noise$};

\draw [-{Stealth}] (0.7, 0.05) to [bend left=25] (proj);

\node at (2.8, -0.6) (midpoint) {$\frac{1}{2}\norm{\vtr{w} - \vtr{x}}$};

\draw [-{Stealth}] (1.5, 0.01) to [bend right=25] (midpoint);

\pic ["$\theta$", draw, thick, -, angle radius = 0.3cm, angle eccentricity=1.5] {angle = xwline--origin--etaline};

\end{tikzpicture}
% \caption{Theorem~\ref{theorem:pert-emb} Illustration}
\caption{}
\label{subfig:orig-vector}
\end{subfigure}
%%%%%%%%%%%%%%%%%%%%%%%%%%%%%%%%%%%%
%%%%%%% End of First figure %%%%%%%%
%%%%%%%%%%%%%%%%%%%%%%%%%%%%%%%%%%%%
\quad
%%%%%%%%%%%%%%%%%%%%%%%%%%%%%
%%%%%%% Second figure %%%%%%%
%%%%%%%%%%%%%%%%%%%%%%%%%%%%%
\begin{subfigure}[t]{0.2\textwidth}
\centering
\begin{tikzpicture}
\draw[thick,-] (0,0) -- (3,0) coordinate[](xline); % horizontal
\draw[thick, dashed, -] (1,0) -- (1,2);
\draw[thick,-] (0,0) -- (1,2) coordinate[](eline);
\draw[thick,-] (3,0) -- (1,2);

\tkzDefPoint(0,0){w};
\tkzLabelPoint[left](w){$\vtr{w}$}
\node at (w)[circle,fill,inner sep=1.5pt]{};

\tkzDefPoint(1,2){eta};
\tkzLabelPoint[left](eta){$\noise$}
\node at (eta)[circle,fill,inner sep=1.5pt]{};

\tkzDefPoint(3,0){x};
\tkzLabelPoint[right](x){$\vtr{x}$}
\node at (x)[circle,fill,inner sep=1.5pt]{};

\node at (0.2,1) (A){$A$};
\node at (1.2,1) (E){$E$};
\node at (2.3,1) (B){$B$};
\node at (0.5,-0.2) (C){$C$};
\node at (2,-0.2) (D){$D$};

\pic ["$\theta$", draw, thick, -, angle radius = 0.3cm, angle eccentricity=1.5] {angle = xline--w--eline};

%\draw[thick] (0.3,0) arc (0:41:0.5);

\end{tikzpicture}
% \caption{geometric interpretation of the result of Theorem~\ref{theorem:pert-emb}}
\caption{}
\label{subfig:geo-illustrate}
\end{subfigure}
%%%%%%%%%%%%%%%%%%%%%%%%%%%%%%%%%%%%%
%%%%%%% End of Second figure %%%%%%%%
%%%%%%%%%%%%%%%%%%%%%%%%%%%%%%%%%%%%%
\quad
%%%%%%%%%%%%%%%%%%%%%%%%%%%%%
%%%%%%% Third figure %%%%%%%%
%%%%%%%%%%%%%%%%%%%%%%%%%%%%%
\begin{subfigure}[t]{0.4\textwidth}
\centering
\begin{tikzpicture}
\draw[thick,->] (0,0) -- (4,0) coordinate[](xyline); % horizontal
\draw[thick,dashed, lightgray, ->] (0.5,0) -- (1.5,2) coordinate[](petaline);
\draw[thick,->] (0.5,-0.5) -- (1.5,1.5) coordinate[](etaline); 

\tkzDefPoint(-0.5,-1.5){o};
\draw[thick,->] (o) -- (0.5,-0.5) coordinate[] (w); % w
\draw[thick,->] (o) -- (0,0) coordinate[] (x); % x
\draw[thick,->] (0, 0) -- (w); % wx
\draw[thick,->] (o) -- (xyline); % y
\draw[thick,->] (xyline) -- (w); % wy

\tkzDefPoint(0.5,0){meet};

\draw[thick,dashed, lightgray, -] (w) -- (meet);
\draw[thick,dashed, lightgray, -] (1.5, 2) -- (1.5,0);

%\draw[very thick,-] (meet) -- (1.5,0); 

\tkzDefPoint(2,0){xy};
\node at (xy)[circle,fill,inner sep=1.5pt]{};

\node at (-0.5, -0.5) {$\vtr{x}$};
\node at (0.3, -1) {$\vtr{w}$};
\node at (1.5, -1.15) {$\vtr{y}$};
\node at (3, 0.25) {$\vtr{y} - \vtr{x}$};
\node at (1, 1) {$\noise$};

%\tkzDefPoint(2.8,1.2){proj};
\node at (3.25,1) (projn) {projection of $\noise$};
\draw [-{Stealth}] (1.25, 0.05) to [bend left=25] (projn);

\node at (-1,1) (projwx) {projection of $\vtr{w} - \vtr{x}$};
\draw [-{Stealth}] (0.25, 0.05) to (projwx);

\node at (3.5,-1) (midpoint) {$\frac{1}{2} \norm{\vtr{y} - \vtr{x}}$};
\draw [-{Stealth}] (xy) to [bend right=25] (midpoint);

\node at (-1.5,-1) (xwlabel) {${\vtr{w} - \vtr{x}}$};
\draw [-{Stealth}] (0.25,-0.25) to [bend left=25] (xwlabel);

\node at (4,-0.5) (wylabel) {${\vtr{w} - \vtr{y}}$};
\draw [-{Stealth}] (1, -0.5) to [bend right=5] (wylabel);

\pic ["$\theta$", draw, thick, -, angle radius = 0.3cm, angle eccentricity=2] {angle = xyline--meet--petaline};

\pic [draw, thick, lightgray, -, angle radius = 0.3cm, angle eccentricity=1] {angle = w--x--xyline};

\end{tikzpicture}
% \caption{Theorem~\ref{theorem:dist-flip} Illustration}
\caption{}
\label{subfig:xvsy}
\end{subfigure}
%%%%%%%%%%%%%%%%%%%%%%%%%%%%%%%%%%%%
%%%%%%% End of Third figure %%%%%%%%
%%%%%%%%%%%%%%%%%%%%%%%%%%%%%%%%%%%%
\caption{Illustrations of Theorems~\ref{theorem:pert-emb} and \ref{theorem:dist-flip}. Subfigure~\subref{subfig:orig-vector} illustrates Theorem~\ref{theorem:pert-emb}, Subfigure~\subref{subfig:geo-illustrate} shows a geometric interpretation of the result of Theorem~\ref{theorem:pert-emb}, and Subfigure~\subref{subfig:xvsy} illustrates Theorem~\ref{theorem:dist-flip}.}
\end{figure*}

Next we are interested in knowing when the nearest neighbor $\vtr{x}$ of the original word $\vtr{w}$ is closer to the noisy embedding $\vtr{w}^*$ than any other neighbor $\vtr{y}$ of $\vtr{w}$. 

\begin{theorem}
\label{theorem:dist-flip}
Let $\noise$ be a noise vector. Let $\vtr{w} \in \mathcal{D}$ be a word embedding, and let $\vtr{w}^* = \vtr{w} + \noise$ be the perturbed embedding. Let $\vtr{x} \in \mathcal{D}$ be the nearest neighbor of $\vtr{w}$. Let $\vtr{y} \in \mathcal{D}$ be any other embedding different from $\vtr{w}$ and $\vtr{x}$. Then $\vtr{w}^*$ is closer to $\vtr{x}$ than $\vtr{y}$ if 
\[
\norm{\vtr{w} - \vtr{x}} \cos \theta_{\vtr{w} - \vtr{x}, \vtr{y} + \vtr{x}} + r \cos \theta_{\vtr{\noise, \vtr{y} - \vtr{x}}} < \frac{1}{2} \norm{\vtr{y} - \vtr{x}}
\]
\end{theorem}
\begin{proof}
See Appendix~\ref{app:proof}.
\end{proof}
An interpretation of the above is as follows: If the sum of projections of the noise vector and that of $\vtr{w} - \vtr{x}$ on $\vtr{y} - \vtr{x}$ is less than half its length, then $\vtr{x}$ will be closer to the noisy embedding than $\vtr{y}$. Figure~\ref{subfig:xvsy} shows a graphical illustration of this result. The result of the theorem can also be explained in a manner similar to Figure~\ref{subfig:geo-illustrate}. We will revisit these results in Section~\ref{sec:loss-function} to see how often they are true for word embeddings. For now we notice that both the results of Theorem~\ref{theorem:pert-emb} and \ref{theorem:dist-flip} involve the product of the length $r$ of the noise vector and the cosine of its angle with a vector in $\mathbb{R}^n$. The following theorem characterizes this distribution. 

\begin{theorem}
\label{theorem:dot-prod}
Let $n \geq 2$. Let $\noise$ be a noise vector. Let $\vtr{w} \in \mathbb{R}^n$ be a non-zero vector. Then 
\begin{equation}
    \dotprod{\vtr{\noise}}{\vtr{w}} = r \norm{\vtr{w}}\cos \theta_{\noise, \vtr{w}},
\end{equation}
where $r \sim f_\text{G}$ given in Eq.~\eqref{eq:gamma-dist} and $k = \cos \theta_{\noise, \vtr{w}} \sim f_\text{B}$ given as
\begin{equation}
\label{eq:beta-dist}
f_{\text{B}}(k) = \frac{1}{B(\frac{n-1}{2}, \frac{1}{2})} (1 - k^2)^{\frac{n-1}{2} - 1}, \quad k \in [-1, 1],
\end{equation}
where $B(\cdot, \cdot)$ is the beta function defined as:
\[
B(a, b) = \frac{\Gamma(a) \Gamma(b)}{\Gamma(a+b)},
\]
for all real numbers $a, b > 0$.
\end{theorem}

\begin{proof}
From the definition of the dot product $\dotprod{\vtr{\noise}}{\vtr{w}}$ equals
\[
\norm{\noise} \norm{\vtr{w}} \cos \theta_{\noise, \vtr{w}} = \norm{r \hat{\vtr{u}}} \norm{\vtr{w}} \cos \theta_{\noise, \vtr{w}} = r \norm{\vtr{w}} \cos \theta_{\noise, \vtr{w}}.
\]
Assuming $\norm{\vtr{w}} \neq 0$, expanding the left hand side, we have:
\begin{align*}
    \sum_{i=1}^n \eta_i w_i &= r \norm{\vtr{w}} \cos \theta_{\noise, \vtr{w}} \\
    \sum_{i=1}^n r \hat{u}_i w_i &= r \norm{\vtr{w}} \cos \theta_{\noise, \vtr{w}} \\
    \sum_{i=1}^n \hat{u}_i \frac{w_i}{\norm{\vtr{w}}} &= \cos \theta_{\noise, \vtr{w}}\\
    \dotprod{\hat{\vtr{u}}}{\hat{\vtr{w}}} &= \cos \theta_{\noise, \vtr{w}},
\end{align*}
where $\hat{\vtr{w}}$ is a unit vector obtained from $\vtr{w}$ by dividing it by its norm. Being unit vectors, the two vectors are on the surface of the $n$-dimensional hypersphere of radius one. Furthermore, $\hat{\vtr{u}}$ is uniformly distributed on the surface of this hypersphere by construction. Since the dot product is rotationally invariant, we can align ${\hat{\vtr{w}}}$ to align with the unit vector $\hat{\vtr{e}}_1$ whose first coordinate is $1$ and all other coordinates are 0. Since $\hat{\vtr{u}}$ is uniformly distributed on the surface of the hypersphere, the rotated vector is still uniformly distributed. Thus, without fear of ambiguity, let us also call this vector $\hat{\vtr{u}}$. Then we see that:
\begin{equation}
\label{eq:unit-dist}
\cos \theta_{\noise, \vtr{w}} = \dotprod{\hat{\vtr{u}}}{\hat{\vtr{e}}_1} = \hat{u}_1 = \frac{u_1}{\norm{\vtr{u}}},   
\end{equation}
where $u_1 \sim \mathcal{N}(0, 1)$ and $\norm{\vtr{u}}$ is the norm of an $n$-dimensional vector each element of which is independently distributed as $\mathcal{N}(0, 1)$. We now find the PDF of the distribution in Eq.~\eqref{eq:unit-dist}.\footnote{The derivation is taken from \url{https://math.stackexchange.com/questions/185298/random-point-uniform-on-a-sphere}. We reproduce it here to add missing details.} 

Let $U_i$ be random variables distributed as $\mathcal{N}(0, 1)$, for $1 \leq i \leq n$. We are interested in the distribution of
\[
K = \frac{U_1}{\sqrt{U^2_1 + U^2_2 + \cdots + U^2_n}}
\]
The range of this variable is in the interval [-1, 1]. Let $-1 < k < 0$. Then we see that
\begin{align}
    &\Pr\left[ \frac{U_1}{\sqrt{U^2_1 + U^2_2 + \cdots + U^2_n}} \leq k \right] \nonumber\\
    &= \Pr\left[ {U^2_1} \geq k^2 ({U^2_1 + U^2_2 + \cdots + U^2_n}) \mid U_1 < 0\right] \nonumber\\
    & = \Pr\left[ \frac{U^2_1}{U^2_2 + \cdots + U^2_n} \geq \frac{k^2}{1 - k^2} \mid U_1 < 0\right] \nonumber\\
    & = \Pr\left[ \frac{(n-1)U^2_1}{U^2_2 + \cdots + U^2_n} \geq \frac{(n-1)k^2}{1 - k^2} \mid U_1 < 0\right] \nonumber\\
    & = \frac{1}{2}\Pr\left[ \frac{(n-1)U^2_1}{U^2_2 + \cdots + U^2_n} \geq \frac{(n-1)k^2}{1 - k^2} \right] \label{eq:f-dist-inter},
\end{align}
where in the last step we have used the fact that $U_1$ is symmetric. Now $U_1^2$ is a chi-squared variable with 1 degree of freedom and $U^2_2 + \cdots + U^2_n$ is a chi-squared variable with $n - 1$ degrees of freedom~\cite[\S 4.3]{mood1950introduction}. Thus, the ratio $\frac{(n-1)U^2_1}{U^2_2 + \cdots + U^2_n}$ is an $F$-distributed random variable with degrees of freedom 1 and $n-1$~\cite[\S 4.4]{mood1950introduction}. The CDF of the $F$-distributed random variable $X$ with $1$ and $n-1$ degrees of freedom is given by:\footnote{See for example: \url{https://mathworld.wolfram.com/F-Distribution.html}}
\begin{equation}
\label{eq:f-dist-cdf}
    F_X\left(x; 1, n-1\right) = I_{\frac{x}{n-1 + x}}\left(\frac{1}{2}, \frac{n-1}{2}\right),
\end{equation}
where $I_y(a, b)$ is the regularized incomplete beta function given as
\[
I_y(a, b) = \frac{B_y(a, b)}{B(a, b)} = \frac{1}{B(a, b)} \int^{y}_{0} t^{a-1} (1-t)^{b-1} \; dt,  
\]
for $a, b > 0$. The function satisfies the relation~\cite[\S 6.4]{1988numerical-recipes}: 
\begin{equation}
\label{eq:I_y}
I_{1-y}(b, a) = 1 - I_y(a, b).
\end{equation}
Now, substituting $x = ((n-1)k^2)/(1 - k^2)$ in $x/(n-1 + x)$ we get 
\[
\frac{x}{n-1 + x} = \frac{(n-1)k^2}{1 - k^2} \frac{1 - k^2}{(n-1)(1 - k^2) + (n-1)k^2} = k^2
\]
Thus, combining this result and using Eqs.~\eqref{eq:f-dist-cdf} and \eqref{eq:I_y}, we get for $-1 < k < 0$
\begin{align}
    \Pr\left[ K \leq k \right] &= \Pr\left[ \frac{U_1}{\sqrt{U^2_1 + U^2_2 + \cdots + U^2_n}} \leq k \right] \nonumber\\
    &= \frac{1}{2}\Pr\left[ \frac{(n-1)U^2_1}{U^2_2 + \cdots + U^2_n} \geq \frac{(n-1)k^2}{1 - k^2} \right] \nonumber\\
    &=  \frac{1}{2}\left(1 - \Pr\left[ \frac{(n-1)U^2_1}{U^2_2 + \cdots + U^2_n} \leq \frac{(n-1)k^2}{1 - k^2} \right]\right) \nonumber\\
    &= \frac{1}{2}\left(1 - F_X\left(x; 1, n-1\right) \right) 
    \nonumber \\
    &=  \frac{1}{2}\left(1 - I_{k^2}\left(\frac{1}{2}, \frac{n-1}{2} \right)\right) \nonumber \\
    &=  \frac{1}{2} I_{1-k^2}\left(\frac{n-1}{2}, \frac{1}{2} \right) \label{eq:f-dist-inter-2}
\end{align}
Now taking the derivative of the integrand in $I_{1-k^2}((n-1)/2, 1/2)$ with respect to $k$ with $-1 < k < 0$, we get:
\begin{align*}
    &\frac{d}{dk} \int^{1-k^2}_{0} t^{(n-1)/2-1} (1-t)^{1/2-1} \; dt \\
    &=  (1 - k^2)^{\frac{n-1}{2} - 1} (1 - (1-k^2))^{1/2} (-2k) \\
    & = - 2 (1 - k^2)^{\frac{n-1}{2} - 1} \frac{k}{\sqrt{k^2}} \\
    & = - 2  (1 - k^2)^{\frac{n-1}{2} - 1} \frac{k}{|k|} \\
    & = 2  (1 - k^2)^{\frac{n-1}{2} - 1} 
\end{align*}
Thus, from Eq.~\eqref{eq:f-dist-inter-2}, if we denote the PDF of $K$ by $f_\text{B}$, we get for $-1 < k < 0$,
\[
    f_B(k) =  \frac{1}{2} \frac{2 (1 - k^2)^{\frac{n-1}{2} - 1} }{B(\frac{n-1}{2}, \frac{1}{2})} = \frac{1}{B(\frac{n-1}{2}, \frac{1}{2})} (1 - k^2)^{\frac{n-1}{2} - 1}.
\]
Finally, since $K$ is symmetric, the above is the PDF for $k \in [-1, 1]$.
\end{proof}

\subsubsection*{Remark}
We call the above, the \emph{noisy dot product} distribution. In light of the above theorem, we call the random variable $R$ distributed as $f_\text{G}$ the \emph{length component} of this distribution and the random variable $K = \cos \theta_{\noise, \vtr{w}}$ distributed as $f_\text{B}$ with respect to \emph{any} word embedding or vector in $\mathbb{R}^n$ as the \emph{angular component} of the distribution, with $Z = RK$ denoting the overall distribution. 

\section{Moments and Tail Bounds of the Noisy Dot Product Distribution}
\label{sec:moments}
In order to rule out any unusual behavior of the noisy dot product distribution $Z = RK$ in higher dimensions, we explore its properties in detail. These properties include its probability density function (PDF), cumulative distribution function (CDF), expectation, variance, and tail bounds of its components, i.e., $R$ and $K$. In the process, we also find an expression for all moments of the component $K$.  

\subsection{CDF and PDF of the Distribution}
\label{subsec:cdf-pdf}
%In both theorems, we deal with the product distribution of the length of the noise vector and the cosine of its angle with an embedding vector. From Theorem~\ref{theorem:dot-prod}, we see that the cosine of the angle follows a certain distribution. Let $R$ denote the random variable denoting the length $r$ of the noise vector, and let $K$ denote the random variable denoting the cosine of the noise vector $\noise$ with any embedding. Let $Z = RK$. Then, 
We are interested in:
\begin{equation*}
    F_Z(z) = \Pr(Z \leq z) = \Pr(RK \leq z) 
\end{equation*}
Now, $R$ and $K$ are independent. Also, the density function of $R$ is non-zero on positive values of $r$ and that of $K$ is non-zero on $-1 \leq k \leq 1$. Furthermore, if $K = z/r$, then $r \geq |z|$, for $K$ to be less than or equal to $1$. Therefore, we get:
\begin{align}
    F_Z(z) &= \Pr[RK \leq z] \nonumber \\
    &=  \int^{\infty}_{-\infty} \Pr[K \leq z/R \mid R = r] f_G (r) dr \nonumber\\
    &=\int^{\infty}_{|z|} \Pr[K \leq z/R \mid R = r] f_G (r) dr \nonumber\\
    &=\int^{\infty}_{|z|} \Pr[K \leq z/r] f_G (r) dr \nonumber\\
    &=\int^{\infty}_{|z|} \left( \int^{z/r}_{-1} f_B(k) dk \right) f_G (r) dr \nonumber\\
    &=\int^{\infty}_{|z|} f_G (r) \left( \int^{z/r}_{-1} f_B(k) dk \right)  dr \nonumber\\
 &= \frac{1}{\Gamma(n)\epsilon^{-n} B(\frac{n-1}{2}, \frac{1}{2})}  \int^{\infty}_{|z|} r^{n-1}e^{-\epsilon r} \nonumber\\
 &\quad\times \left( \int^{z/r}_{-1}  \left(1 - k^2\right)^{\frac{n-1}{2} - 1} dk \right) dr \label{eq:cdf-z}
\end{align}
Taking the derivative of the above with respect to $z$ using the fundamental theorem of calculus gives us the PDF of this distribution:
\begin{align}
    f_Z(z) &= \frac{d}{dz} 
     \int^{\infty}_{|z|} f_G (r) \left( \int^{z/r}_{-1} f_B(k) dk \right)  dr \nonumber\\
     &= \frac{d}{dz} 
     \int^{\infty}_{0} f_G (r) \left( \int^{z/r}_{-1} f_B(k) dk \right)  dr \nonumber\\
     &=  
     \int^{\infty}_{0} f_G (r) f_B(z/r) \frac{1}{r}  dr \label{eq:prod-ind}\\
     &=  
     \int^{\infty}_{|z|} f_G (r) f_B(z/r) \frac{1}{r}  dr \label{eq:pdf-z}
\end{align}
The last equality follows since $f_\text{B}(z/r) = 0$ for $r < |z|$. Now, for any $\delta > 0$, we have that
\begin{align*}
    f_Z(\delta) &= \int^{\infty}_{|\delta|} f_G (r) f_B(\delta/r) \frac{1}{r}  dr \\
    &= \int^{\infty}_{|-\delta|} f_G (r) f_B(-\delta/r) \frac{1}{r}  dr  = f_Z(-\delta),
\end{align*}
where the second step follows since $f_\text{B}$ is symmetric around 0. Thus, the distribution is symmetric around 0. As we shall see in Section~\ref{sec:moments-k}, the expected value of $Z$ is 0. Thus, the distribution is symmetric around its mean. 
% Similarly, it is straightforward to show that 
% \begin{equation}
% \label{eq:1-cdf-z}
%     \Pr[RK \geq z] 
%     = \frac{1}{\Gamma(n)\epsilon^{-n} B(\frac{n-1}{2}, \frac{1}{2})}  \int^{\infty}_{|z|} r^{n-1}e^{-\epsilon r} \left( \int^{1}_{z/r}  \left(1 - k^2\right)^{\frac{n-1}{2} - 1} dk \right) dr  
% \end{equation}
Unfortunately, the integral above does not have an easy analytical solution. We can, however, numerically evaluate it or through Monte Carlo simulations by repeatedly sampling the noise vector.  

\subsection{Moments of the Angular Component}
\label{sec:moments-k}
Let $K$ denote the random variable distributed as Eq.~\eqref{eq:beta-dist}. We are interested in the moment $\mathbb{E}[K^j]$ of this distribution with $j \geq 0$. We parameterize this distribution by using $K_n$ to denote the random variable $K$ with a given value of $n \geq 2$.

\begin{theorem}
\label{theorem:beta-moments}
Let $K \sim f_\text{B}$ as defined in Eq.~\eqref{eq:beta-dist}. Let $K_n$ denote $K$ for a particular value of $n$. Let $\mu(j, n)$ denote the $j$th moment of $K_n$. Then, for all $n \geq 2$
\begin{equation}
\label{eq:moments-beta}
    \mu(j, n) = \begin{cases}
        1, \; &\text{ if } j = 0,\\
        0, \; &\text{ if } j \text{ is odd},\\
        \frac{(n-2)!! (j-1)!!}{(n-2+j)!!}, \; &\text{ if } j \text{ is even}
    \end{cases}
\end{equation}
In particular, $\mathbb{E}[K] = 0$ and $\text{Var}[K] = \frac{1}{n}$. 
\end{theorem}
\begin{proof}
See Appendix~\ref{app:proof}.
\end{proof}

\subsection{Tail Bounds, Expectation and Variance} 
\label{subsec:tail-bounds}

The following definition and  the follow-up theorem are taken from~\cite[\S 2]{wainwright-hds}. 
\begin{definition}[Sub-Gaussian Random Variable]
\label{def:sg}
A random variable $X$ with mean $\mu = \mathbb{E}[X]$ is sub-Gaussian if there is a positive number $\sigma$ such that 
\[
\mathbb{E}\left[ e^{\lambda (X - \mu)} \right] \leq e^{\sigma^2\lambda^2/2}, \; \text{ for all } \lambda \in \mathbb{R}.
\]
\end{definition}
Here $\sigma$ is called the sub-Gaussian parameter.

\begin{theorem}[Sub-Gaussian Tail Bound]
\label{theorem:sg-tail}
A sub-Gaussian random variable $X$ with mean $\mu = \mathbb{E}[X]$ and sub-Gaussian parameter $\sigma$ satisfies
\[
\Pr \left[X - \mu \geq t \right] \leq e^{-\frac{t^2}{2\sigma^2}}, \text{ and } \Pr \left[X - \mu \leq -t \right] \leq e^{-\frac{t^2}{2\sigma^2}},
\]
and combining the two
\[
\Pr \left[|X - \mu| \geq t \right] \leq 2 e^{-\frac{t^2}{2\sigma^2}}, \; \text{ for all } t \in \mathbb{R}.
\]
\end{theorem}
We prove the following result to be used for the next theorem.
\begin{lemma}
\label{lemma:lin-exp}
Let $K \sim f_B$. Then for any $\lambda \in \mathbb{R}$ we have
\[
\mathbb{E}\left[
    \sum_{j = 0}^\infty \frac{(\lambda K)^j}{j !} \right] = 
    \sum_{j = 0}^\infty \frac{\lambda^j \mathbb{E}[K^j]}{j !} 
\]
\end{lemma}
\begin{proof}
Since $-1 \leq K \leq 1$, we have $|K| \leq 1$. Therefore $\mathbb{E}[|K|^j] \leq 1 $. We have, for any $j \geq 0$:
\[
\mathbb{E}\left[
   \frac{\lambda^j |K|^j}{j !} \right] = \frac{\lambda^j}{j!} \mathbb{E}[|K|^j] \leq  \frac{\lambda^j}{j!}.
\] 
For any integer $m$, let
\[
S_m = \mathbb{E}\left[ \sum_{j = 0}^m
   \frac{\lambda^j |K|^j}{j !} \right] = \sum_{j = 0}^m \frac{\lambda^j}{j!} \mathbb{E}[|K|^j] \leq \sum_{j = 0}^m \frac{\lambda^j}{j!},
\]
where the last inequality follows from the result above. Furthermore, the sequence $S_m$ is monotonically increasing if $\lambda \geq 0$ or monotonically decreasing if $\lambda < 0$, since $\mathbb{E}[|K|^j] \geq 0$ as $|K| \geq 0$. Thus, $S_m$ is a monotone sequence. From the Taylor series expansion of the exponential function, we have for any $\lambda \in \mathbb{R}$
\begin{align*}
    e^{|\lambda|} = \sum_{j = 0}^\infty \frac{|\lambda|^j}{j!} \geq \sum_{j = 0}^m \frac{|\lambda|^j}{j!} \geq \sum_{j = 0}^m \frac{\lambda^j}{j!} \geq S_m.
\end{align*}
Thus, $S_m$ is bounded. From the monotone convergence theorem~\cite[\S 2.4]{abbott-analysis}, $S_m$ converges. Therefore, the statement of the theorem follows as expectation is linear in this case~\cite[\S 2.1.1]{mitzenmacher2017probability}.   
\end{proof}

\begin{theorem}
\label{theorem:beta-sg}
Let $K \sim f_B$ where $f_B$ is as defined in Eq.~\ref{eq:beta-dist}. Then $K$ is sub-Gaussian with parameter $\sigma = \frac{1}{\sqrt{n}}$. 
\end{theorem}
\begin{proof}
See Appendix~\ref{app:proof}. 
\end{proof}

\begin{corollary}
\label{cor:beta-tail-bound}
Let $K \sim f_B$. Then for any $c \in \mathbb{R}$,
\[
\Pr\left[|K| \geq \frac{c}{\sqrt{n}}\right] \leq 2 e^{-\frac{c^2}{2}}.
\]
\end{corollary}

Next we prove generic lower and upper tail bounds for the gamma distributed random variable $R$.\footnote{These results are generalizations of the result for the upper tail bound with $c = 2$ in \url{https://math.hawaii.edu/~grw/Classes/2013-2014/2014Spring/Math472_1/Solutions01.pdf}.}

\begin{theorem}
\label{theorem:gamma-tail-bound}
Let $R \sim f_G$. Then for any real number $c > 1$,
\[
\Pr\left[R \geq \frac{cn}{\epsilon}\right] \leq \left( \frac{c}{e^{c-1}} \right)^n,
%\]    
\text{ and } 
%\[
\Pr\left[R \leq \frac{n}{c \epsilon}\right] \leq \frac{1}{(ce^{(1-c)/c})^n}
\]  
\end{theorem}
\begin{proof}
See Appendix~\ref{app:proof}.
\end{proof}

Combining the result from Corollary~\ref{cor:beta-tail-bound} and Theorem~\ref{theorem:gamma-tail-bound}, we have the following bound on the overall noise distribution.
\begin{theorem}
\label{theorem:tail-bound-z}
Let $Z = RK$. Then, for all $c_1, c_2 \in \mathbb{R}$, where $c_2 > 1$, we have
\[
\Pr\left[ |Z| \leq \frac{c_1c_2 \sqrt{n}}{\epsilon}\right] \geq 2\left(1 - e^{\frac{-c_1^2}{2}}\right) \left(1 - \left( \frac{c_2}{e^{c_2-1}} \right)^n\right) - 1
\]
\end{theorem}
\begin{proof}
We have
\begin{align*}
    \Pr\left[ Z \leq \frac{c_1c_2 \sqrt{n}}{\epsilon}\right] &= \Pr\left[ RK \leq \frac{c_1c_2 \sqrt{n}}{\epsilon}\right] \\
    &= \Pr\left[ RK \leq \frac{c_1}{\sqrt{n}}\frac{c_2 n}{\epsilon}\right] \\
    &\geq \Pr\left[ K \leq \frac{c_1}{\sqrt{n}}\right] \Pr\left[R \leq  \frac{c_2 n}{\epsilon}\right] \\
    &\geq \left(1 - e^{\frac{-c_1^2}{2}}\right) \left(1 - \left( \frac{c_2}{e^{c_2-1}} \right)^n\right), 
\end{align*}
where the last inequality follows from Corollary~\ref{cor:beta-tail-bound} and Theorem~\ref{theorem:gamma-tail-bound}. From Section~\ref{subsec:cdf-pdf}, $Z$ is symmetric, and hence the above bound is also true for $\Pr [ Z \geq - c_1 c_2 \sqrt{n}/\epsilon]$. Through Bonferroni's inequality
\begin{align*}
    &\Pr\left[ |Z| \leq \frac{c_1c_2 \sqrt{n}}{\epsilon}\right] \\
    &= \Pr\left[ Z \leq \frac{c_1c_2 \sqrt{n}}{\epsilon} \text{ and } Z \geq -\frac{c_1c_2 \sqrt{n}}{\epsilon}\right] \\
    &\geq \Pr\left[ Z \leq \frac{c_1c_2 \sqrt{n}}{\epsilon} \right] + \Pr\left[  Z \geq -\frac{c_1c_2 \sqrt{n}}{\epsilon}\right] - 1\\
    &\geq 2\left(1 - e^{\frac{-c_1^2}{2}}\right) \left(1 - \left( \frac{c_2}{e^{c_2-1}} \right)^n\right) - 1,
\end{align*}
as required.
\end{proof}
As an illustration of this inequality, with $n = 100$ and $\epsilon = 10$, more than 99 percent of the probability mass of $Z$ lies within the interval $\pm 4.8\sqrt{n}/\epsilon = 4.8$ (with $c_1 \approx 3.46$ and $c+2 \approx 1.39$). On the other hand, with $n = 10$, more than 99 percent of the probability mass of $Z$ lies within the interval $\pm 8.74\sqrt{n}/\epsilon \approx 2.76$ (put for example $c_1 \approx 3.46$ and $c_2 \approx 2.53$).\footnote{These values were obtained by fixing the bound from Theorem~\ref{theorem:tail-bound-z} at $0.99$ and numerically finding the constants $c_1$ and $c_2$ by keeping the two terms in the product equal to each other.} Thus, the majority of the mass of $Z$ is concentrated within $\mathcal{O}(\sqrt{n}/\epsilon)$.

Finally we have the following theorem on the expected value and variance of $Z$, together with the convergence to the expected value if we sample a large number of instances of $Z$. 
\begin{theorem}
\label{theorem:exp-z}
Let $Z = RK$. Then $\mathbb{E}[Z] = 0$ and $\text{Var}[Z] = \frac{n+1}{\epsilon^2}$. 
\end{theorem}
\begin{proof}
We know that the gamma distributed random variable $R$ has $\mathbb{E}[R] = n/\epsilon$ and $\text{Var}[R] = n/\epsilon^2$
Since $R$ and $K$ are independent, we have $\mathbb{E}[Z] = \mathbb{E}[RK]= \mathbb{E}[R] \mathbb{E}[K] = \frac{n}{\epsilon} \cdot 0 = 0$. Now using this result, and again because $R$ and $K$ are independent, we have:
\begin{align*}
    \text{Var}[Z] &= \mathbb{E}[Z^2] - (\mathbb{E}[Z])^2 = \mathbb{E}[Z^2]\\
    &= \mathbb{E}[R^2 K^2]
    = \mathbb{E}[R^2] \mathbb{E}[K^2]\\
    &= ( \text{Var}[R] + (\mathbb{E}[R])^2) \text{Var}[K]\\
    &= \left( \frac{n}{\epsilon^2} + \frac{n^2}{\epsilon^2} \right)\left( \frac{1}{n}\right) = \frac{n+1}{\epsilon^2}
\end{align*}
\end{proof}

\begin{corollary}
\label{cor:z-chebyshev}
Let $Z_1, Z_2, \ldots, Z_m$ be $m$ independent samples of the random variable $Z = RK$. Let $\bar{Z} = \frac{1}{m}\sum_{i=1}^m Z_i $. Then for any $\delta > 0$
\begin{equation}
\label{eq:z-chebyshev}
\Pr\left[ |\bar{Z}| > \delta \right] \leq \frac{n+1}{m\epsilon^2 \delta^2}
\end{equation}
\end{corollary}
\begin{proof}
The result is obtained by putting the expected value and variance of $Z$ in Chebyshev's inequality~\cite[\S 2]{wainwright-hds}. 
\end{proof}
The above result shows that the higher the dimension and/or the lower the values of $\epsilon$ the slower will be the convergence of the average of $Z$ to $0$, the expectation. Thus, the behaviour of the noisy dot product distribution is as we would expect: it is symmetric with mean 0, and it converges to this expectation inversely proportional to $\epsilon$, the privacy parameter. 

\section{The Loss Function and Consequences}
\label{sec:loss-function}
Equipped with the results of the last section, we can now explore in depth how the nearest neighbor is obtained via the post-processing step. Recall the objective function from Eq.~\eqref{eq:min-problem}. Taking its square, we get
\begin{align*}
    \norm{\vtr{w}^* - \vtr{x}}^2 &= \norm{\vtr{w} + \noise - \vtr{x}}^2 \\
    &= \norm{\vtr{w} - \vtr{x}}^2 + \norm{\noise}^2 + 2 \dotprod{\vtr{w} - \vtr{x}}{\noise}\\
    &= \norm{\vtr{w} - \vtr{x}}^2 + r^2 + 2 \dotprod{\vtr{w} - \vtr{x}}{\noise} \\
    &= \norm{\vtr{w} - \vtr{x}}^2 + r^2 + 2 \dotprod{\vtr{w}}{\noise} - 2 \dotprod{\vtr{x}}{\noise}
\end{align*}
Now, the terms $\dotprod{\vtr{w}}{\noise}$ and $r^2$ are the same for all $\vtr{x} \in \mathcal{D}$, and therefore we can ignore them when finding the minimum. Let us define the \emph{loss function} containing the remaining terms for all $\vtr{x} \in \mathcal{D}$ as
\begin{equation}
\label{eq:loss-function}
L(\vtr{x}) = \norm{\vtr{w} - \vtr{x}}^2 - 2 \dotprod{\vtr{x}}{\noise} = \norm{\vtr{w} - \vtr{x}}^2 + 2 \norm{\vtr{x}} r \cos \theta_{\noise, \vtr{x}}
\end{equation}
Indeed, using partial derivatives we can see that the solution that minimizes the loss function $L$ is $\vtr{x} = \vtr{w} + \noise = \vtr{w}^*$, but we know that with overwhelming probability this vector is not part of the embedding space. We have the following result.
\begin{theorem}
\label{theorem:loss-function}
Let $L$ be defined as in Eq.~\eqref{eq:loss-function}. Then for any $\vtr{x} \in \mathcal{D}$
\[
\mathbb{E}[ L(\vtr{x}) ] = \norm{\vtr{w} - \vtr{x}}^2.
\]
In particular, $\mathbb{E}[ L(\vtr{w}) ] = 0 $, and $\mathbb{E}[ L(\vtr{x}) ] > 0$ iff $\vtr{x} \neq \vtr{w}$. Furthermore, for any $\vtr{x}, \vtr{y} \in \mathcal{D}$ with $\vtr{x} \neq \vtr{y}$ and $\norm{\vtr{w} - \vtr{x}} < \norm{\vtr{w} - \vtr{y}}$, we have 
\[
\Pr\left[ L(\vtr{x}) < L(\vtr{y}) \right] > 1/2.
\]
\end{theorem}
\begin{proof}
From Eq.~\eqref{eq:loss-function} and Theorem~\ref{theorem:exp-z}, we see that 
\begin{align*}
\mathbb{E}[ L(\vtr{x}) ] &= \mathbb{E}\left[ \norm{\vtr{w} - \vtr{x}}^2 + 2 \norm{\vtr{x}} r \cos \theta_{\noise, \vtr{x}} \right]  \\
&= \norm{\vtr{w} - \vtr{x}}^2 + 2 \norm{\vtr{x}} \mathbb{E} [r \cos \theta_{\noise, \vtr{x}}] \\
&= \norm{\vtr{w} - \vtr{x}}^2,
\end{align*}
from which it follows that the expectation is $0$ only if $\vtr{x} = \vtr{w}$, and otherwise it is greater than 0. For the second part of the theorem, we see that
\begin{align*}
     L(\vtr{y}) - L(\vtr{x}) &= \norm{\vtr{w} - \vtr{y}}^2 - 2 \dotprod{\vtr{y}}{\noise} - \norm{\vtr{w} - \vtr{x}}^2 + 2 \dotprod{\vtr{x}}{\noise} \\
     &=\norm{\vtr{w} - \vtr{y}}^2 - \norm{\vtr{w} - \vtr{x}}^2 + 2 \dotprod{\vtr{x} - \vtr{y}}{\noise}\\
     &= \norm{\vtr{w} - \vtr{y}}^2 - \norm{\vtr{w} - \vtr{x}}^2 + 2 \norm{\vtr{x} - \vtr{y}} r \cos \theta_{\noise, \vtr{x} - \vtr{y}}
\end{align*}
Define the random variable $S$ as 
\[
S = \norm{\vtr{w} - \vtr{y}}^2 - \norm{\vtr{w} - \vtr{x}}^2 + 2 \norm{\vtr{x} - \vtr{y}} Z,
\]
which follows from the fact that $r \cos \theta_{\noise, \vtr{x} - \vtr{y}}$ is distributed as $Z$. Note that from the analysis above, we have $\mathbb{E}[S] = \norm{\vtr{w} - \vtr{y}}^2 - \norm{\vtr{w} - \vtr{x}}^2$. Now, since $\norm{\vtr{w} - \vtr{x}} < \norm{\vtr{w} - \vtr{y}}$, from Eq.~\eqref{eq:cdf-z}, we have that
\begin{align}
    &\Pr[L(\vtr{x}) < L(\vtr{y}) ] = \Pr[S > 0] \nonumber\\
    &= \Pr \left[ Z > \frac{\norm{\vtr{w} - \vtr{x}}^2 - \norm{\vtr{w} - \vtr{y}}^2  }{2 \norm{\vtr{x} - \vtr{y}}}\right] \label{eq:condition}\\
    &> \Pr[Z \geq 0] \nonumber\\
    &= \frac{1}{\Gamma(n)\epsilon^{-n} B(\frac{n-1}{2}, \frac{1}{2})}  \int^{\infty}_{|0|} r^{n-1}e^{-\epsilon r} \nonumber \\
    &\times \left( \int^{0/r}_{-1}  \left(1 - k^2\right)^{\frac{n-1}{2} - 1} dk \right) dr \nonumber\\
    &= \frac{1}{2} \frac{1}{\Gamma(n)\epsilon^{-n}}  \int^{\infty}_{0} r^{n-1}e^{-\epsilon r} dr  = \frac{1}{2} \nonumber,
\end{align}
as required. The second-to-last equality follows because the distribution $f_\text{B}$ is symmetric around 0. 
\end{proof}
This theorem shows that, ignoring the actual difference in probabilities, there is nothing unusual in the nearest neighbor search for any value of $\epsilon$ and $n$: the original word is expected to be output more often, followed by its nearest neighbor, and so on. However, the issue is with the exact values of these probabilities, as we shall see next. 

Eq.~\eqref{eq:condition} gives us another insight. Suppose $\vtr{x}_1$ and $\vtr{x}_2$ are the first two neighbors of $\vtr{w}$. Then first we see that for $i = 1, 2$
\begin{align}
    \Pr[L(\vtr{w}) < L(\vtr{x}_i) ] &=
    \Pr \left[ Z > \frac{\norm{\vtr{w} - \vtr{w}}^2 - \norm{\vtr{w} - \vtr{x}_i}^2  }{2 \norm{\vtr{w} - \vtr{x}_i}}\right]\nonumber\\
    &= \Pr\left[ Z > - \frac{\norm{\vtr{w} - \vtr{x}_i}}{2}\right] \nonumber\\
    &= \Pr\left[ Z \leq \frac{\norm{\vtr{w} - \vtr{x}_i}}{2}\right]
    \label{eq:w-condition},
\end{align}
where the last equality follows from the fact that $Z$ is symmetric around 0 (see Section~\ref{subsec:cdf-pdf}). Compare this to the following where we also use symmetry of the distribution:
\begin{align}
    \Pr[L(\vtr{x}_1) < L(\vtr{x}_2) ] &= \Pr \left[ Z > \frac{\norm{\vtr{w} - \vtr{x}_1}^2 - \norm{\vtr{w} - \vtr{x}_2}^2  }{2 \norm{\vtr{x}_1 - \vtr{x}_2}}\right]  \nonumber\\
    &= \Pr \left[ Z \leq \frac{\norm{\vtr{w} - \vtr{x}_2}^2 - \norm{\vtr{w} - \vtr{x}_1}^2  }{2 \norm{\vtr{x}_1 - \vtr{x}_2}}\right]    
    \label{eq:x1-x2-condition}
\end{align}
Eqs~\eqref{eq:w-condition} and \eqref{eq:x1-x2-condition} are the same equations we derived in Theorem~\ref{theorem:pert-emb} and Eq.~\eqref{eq:cond-2} in Theorem~\ref{theorem:dist-flip}, but this time with the loss function formulation. 

\subsection{Probabilistic Interpretation} 
Eqs~\eqref{eq:w-condition} and \eqref{eq:x1-x2-condition} can be interpreted as upper bounds on the probability that a particular word will be chosen as the nearest neighbor. To see this let $N$ denote the size of the vocabulary. Given a word $\vtr{w}$, we index the vocabulary, as follows: $\vtr{x}_0 = \vtr{w}$ and $\vtr{x}_i$ denotes the $i$th nearest neighbor of $\vtr{w}$ where ties are broken arbitrarily. Let $C_{i, j}$ be the event that word embedding $\vtr{x}_i$ is closer to the noisy embedding $\vtr{w}^*$ than word embedding $\vtr{x}_j$. Furthermore, let $C_{i}$ be the event that word embedding $\vtr{x}_i$ is the nearest neighbor of $\vtr{w}^*$. From Eq.~\eqref{eq:w-condition}, we have the following condition:
\begin{align}
    \Pr[C_i] &= 1 - \Pr[\bar{C}_i] \nonumber\\
    &= 1 - 
    \Pr\left[ \bigcup_{j=0, j \neq i}^N \bar{C}_{i, j}\right] \nonumber\\
%    &\geq 1 - \sum_{j=0, j \neq i}^N \Pr[C_{i, j}]\\
    &\leq 1 - \Pr[\bar{C}_{i, 0}]\nonumber\\
    &= 1 - \Pr[L(\vtr{w}) < L(\vtr{x}_i) ]\nonumber\\
    &= 1 - \Pr\left[ Z \leq 
    \frac{\norm{\vtr{w} - \vtr{x}_i}}{2}\right]\nonumber\\
    &\leq 1 - \Pr\left[ Z \leq \frac{\norm{\vtr{w} - \vtr{x}_1}} {2}\right], \label{eq:condition-w}
\end{align}
where we have used the fact that $\bar{C}_{i, 0} \subseteq \cup_{j=0, j \neq i}^N \bar{C}_{i, j}$, and therefore the probability of former is less than or equal to the probability of the latter. Likewise, we have
\begin{align}
    \Pr[C_i] &\leq 1 - \Pr[\bar{C}_{i, j}]
            = \Pr[{C}_{i, j}] = \Pr[L(\vtr{x}_i) < L(\vtr{x}_j)]\nonumber\\
             &= \Pr \left[ Z \leq \frac{\norm{\vtr{w} - \vtr{x}_j}^2 - \norm{\vtr{w} - \vtr{x}_i}^2  }{2 \norm{\vtr{x}_i - \vtr{x}_j}}\right]  \label{eq:condition-xi},
\end{align}
for $i \neq 0$. The probabilities in Eq.~\eqref{eq:condition-w} and \eqref{eq:condition-xi} are upper bounds on the necessary condition for a word $\vtr{x}_i$ to be chosen as the nearest neighbor to the noisy embedding. 

\subsection{Results on Word Embeddings}
To show how likely the original word $\vtr{w}$ is chosen over its nearest neighbor $\vtr{x}_1$, and contrast it with how likely the nearest neighbor $\vtr{x}_1$ is chosen over the second nearest neighbor $\vtr{x}_2$, we take the right hand quantities bounding the probability mass of the random variable $Z$ in Eqs~\eqref{eq:w-condition} and \eqref{eq:x1-x2-condition}. To ease notation, we denote these quantities by $z_{\vtr{w}, \vtr{x}_1}$ and $z_{\vtr{x}_1, \vtr{x}_2}$, respectively. For each vocabulary, i.e., embedding model, we take the average of these quantities by taking 5,000 random words. The results are shown in Table~\ref{tab:rhs-quantities}. The quantity $z_{\vtr{x}_1, \vtr{x}_{101}}$ is the right hand quantity in Eq~\eqref{eq:x1-x2-condition} for the nearest neighbor and 101th neighbor of $\vtr{w}$.

\begin{table}[!ht]
\centering
% \resizebox{.9\line}{!}{
 \begin{tabular}{c|c|c|c|c} 
 \hline
Vocabulary & $n$ & $z_{\vtr{w}, \vtr{x}_1}$ & $z_{\vtr{x}_1, \vtr{x}_2}$ & $z_{\vtr{x}_1, \vtr{x}_{101}}$\\
 \hline\hline
\verb+GloVe-Twitter+ & 25 & 1.078 & 0.132 & 0.692\\
& 50 & 1.571 & 0.169 & 0.747\\
& 100 & 2.166 & 0.197 & 0.764\\
& 200 & 2.708 & 0.220 & 0.683 \\
\hline
\verb+GloVe-Wiki+ & 50 & 1.369 & 0.172 & 0.757 \\
& 100 & 1.403 & 0.152 & 0.616 \\
& 200 & 2.15 & 0.243 & 0.823 \\
& 300 & 2.896 & 0.282 & 0.999\\
\hline 
\verb+Word2Vec+ & 300 & 0.763 & 0.058 & 0.253\\
\hline
\verb+fastText+ & 300 & 1.493 & 0.239 & 0.912\\ \hline
\end{tabular}
% }
\caption{The right hand side terms  $z_{\vtr{w}, \vtr{x}_1}$,  $z_{\vtr{x}_1, \vtr{x}_2}$ and $z_{\vtr{x}_1, \vtr{x}_{101}}$ as they appear in Eqs.~\eqref{eq:w-condition} and \eqref{eq:x1-x2-condition} for different vocabularies.}
\label{tab:rhs-quantities}
\end{table}

The first thing to notice is that for any given dimension $z_{\vtr{w}, \vtr{x}_1}$ is considerably greater than $z_{\vtr{x}_1, \vtr{x}_2}$. This is because the distance of any word to its first neighbor in high dimensions is considerably higher than the (relative) difference in distances to its first two neighbors. The other observation is more of an illusion. The value of $z_{\vtr{w}, \vtr{x}_1}$ is increasing as we increase the dimensions, by looking at the \verb+GloVe+ vocabularies. Thus it may appear that the problem is exacerbated as we increase the dimensions from say 25 to 200. However, this is not exactly correct. Recall from Theorem~\ref{theorem:tail-bound-z} that the bulk of the mass of $Z$ is within $\mathcal{O}(\sqrt{n}/\epsilon)$. Thus, under two different dimensions $n_1$ and $n_2$, we expect $F_Z(z_1) \approx F_Z(z_2)$ if $\frac{z_1}{z_2} \approx \sqrt{\frac{n_1}{n_2}}$. By looking at the table, we see that this ratio across the two types of \verb+GloVe+ is more or less obeyed by the quantities 
$z_{\vtr{w}, \vtr{x}_1}$ as we cycle through the dimensions, and slightly less so by the quantities 
$z_{\vtr{x}_1, \vtr{x}_2}$. Thus, we might see the problem as being slightly more worsened in higher dimensions but not by much. 

This is more obvious as we plot the probabilities $F_Z(z_{\vtr{w}, \vtr{x}_1})$ and $F_Z(z_{\vtr{x}_1, \vtr{x}_2})$ in Figure~\ref{fig:vocab-cdf}. Although the CDF can be obtained by numerically integrating the integral in Eq.~\eqref{eq:cdf-z}, however, we found the integration to be slow using for example \verb+SymPy+.\footnote{See \url{https://www.sympy.org/}.} We therefore obtained the CDF of $Z$ using Monte Carlo simulations by randomly sampling the noise vector $10,000$ times and then finding the proportion of times it falls below $z_{\vtr{w}, \vtr{x}_1}$ or $z_{\vtr{x}_1, \vtr{x}_2}$. This can be done by sampling $R$ via the Gamma distribution $f_\text{G}$, and then sampling $K$ by checking the cosine of the angle of the noise vector against a fixed vector, e.g., $\hat{\vtr{e}}_1 = (1, 0, 0, \ldots, 0)$. 

As is evident from the figure, $F_Z(z_{\vtr{w}, \vtr{x}_1})$ dominates when $\epsilon > 1$. This means, that the probability of choosing the original word increases, and from Eq.~\eqref{eq:condition-w}, the probability that the nearest neighbor would be output decreases. Furthermore, around $\epsilon = 10$, where the original word is still not entirely certain to be output, the probability $F_Z(z_{\vtr{x}_1, \vtr{x}_{101}})$ is not overwhelming enough (not plotted) to ensure that the nearest neighbor would be output more than the 101th neighbor. Hence far away neighbors are still being output at these $\epsilon$ values, until the original word completely dominates as we further increase $\epsilon$.

% \begin{figure}[!ht]
%     \centering
%     \begin{subfigure}[t]{0.49\linewidth}
%         \includegraphics[width=\linewidth]{figs/glove-twitter-cdf.pdf}
%         \caption{\texttt{GloVe-Twitter}}
%         \label{subfig:glove-twitter-cdf}
%     \end{subfigure}
%     % \hspace{10mm}%
%     \hfill
%     \begin{subfigure}[t]{0.49\linewidth}
%         \includegraphics[width=\linewidth]{figs/glove-wiki-cdf.pdf}
%         \caption{\texttt{GloVe-Wiki}}
%         \label{subfig:glove-wiki-cdf}
%     \end{subfigure}
%     \caption{The probabilities $F_Z(z_{\vtr{w}, \vtr{x}_1})$ (solid) and $F_Z(z_{\vtr{x}_1, \vtr{x}_2})$ (dashed) where $z_{\vtr{w}, \vtr{x}_1}$ and $z_{\vtr{x}_1, \vtr{x}_2}$ are as given in Table~\ref{tab:rhs-quantities} for the two \texttt{GloVe} vocabularies.} 
%     \label{fig:vocab-cdf}
% \end{figure}

\begin{figure}[!ht]
    \centering
    \begin{subfigure}[t]{\linewidth}
        \includegraphics[width=\linewidth]{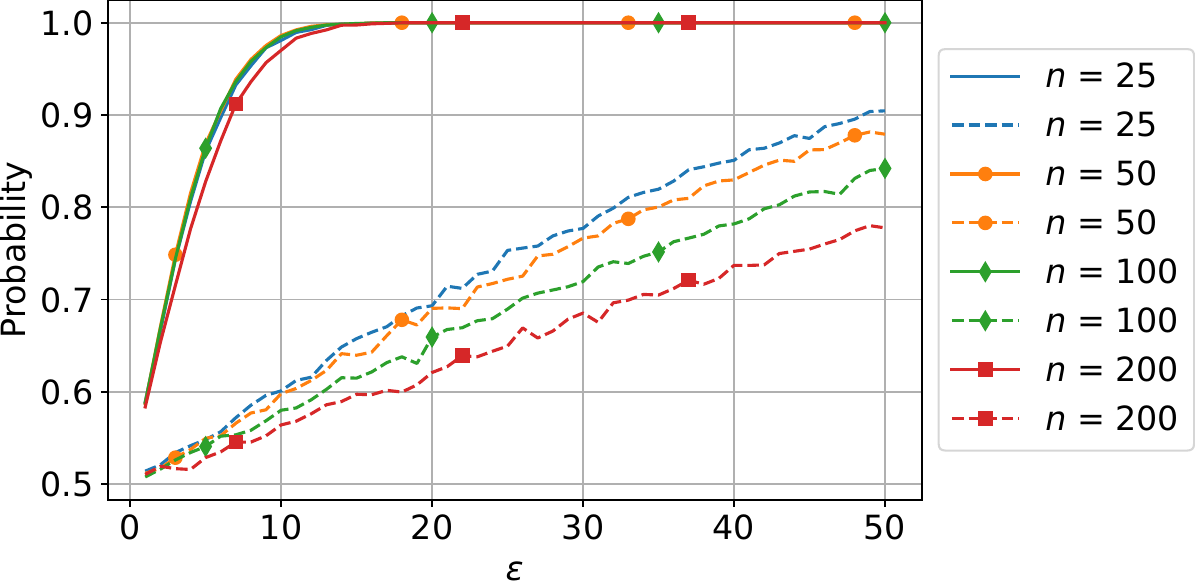}
        \caption{\texttt{GloVe-Twitter}}
        \label{subfig:glove-twitter-cdf}
    \end{subfigure}
    % \hspace{10mm}%
    % \hfill
    \begin{subfigure}[t]{\linewidth}
        \includegraphics[width=\linewidth]{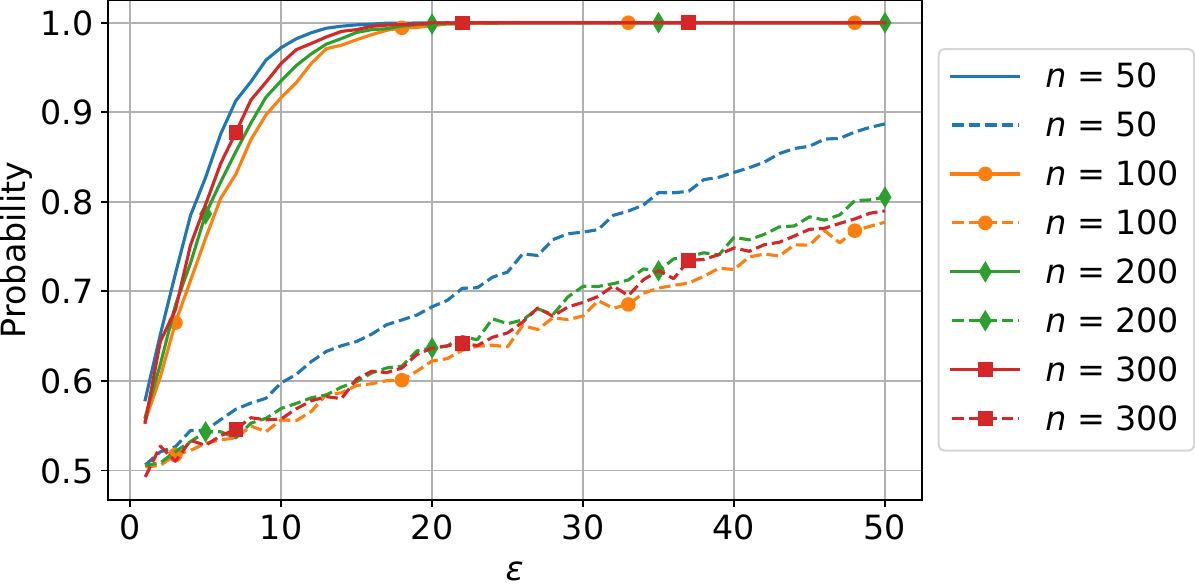}
        \caption{\texttt{GloVe-Wiki}}
        \label{subfig:glove-wiki-cdf}
    \end{subfigure}
    \caption{The probabilities $F_Z(z_{\vtr{w}, \vtr{x}_1})$ (solid) and $F_Z(z_{\vtr{x}_1, \vtr{x}_2})$ (dashed) where $z_{\vtr{w}, \vtr{x}_1}$ and $z_{\vtr{x}_1, \vtr{x}_2}$ are as given in Table~\ref{tab:rhs-quantities} for the two \texttt{GloVe} vocabularies.} 
    \label{fig:vocab-cdf}
\end{figure}

\section{Proposed Fix}
In light of the discussion in the previous section, if we can somehow make the distance of the original word to its nearest neighbor similar to distances between its consecutive neighbors, the issue could be resolved. To do so, one may be tempted to employ the following mechanism. Let $\vtr{w}$ be the original word. We find all nearest neighbors of $\vtr{w}$ according to the Euclidean distance, and assign the function $d_\text{NN}(\vtr{w}, \vtr{x}) = i$ if $\vtr{x}$ is the $i$th nearest neighbor of $\vtr{w}$. We have $d_\text{NN}(\vtr{w}, \vtr{w}) = 0$. We can then sample a word $\vtr{x}$ proportional to $\exp(-\epsilon d_\text{NN}(\vtr{w}, \vtr{x}))$. However, $d_\text{NN}$ is not a metric as it does not satisfy the properties of symmetry and triangle inequality as illustrated in Figure~\ref{fig:nn-not-metric}. In the figure we consider $\mathbb{R}^2$. We have $d_\text{NN}(\vtr{x}_1, \vtr{x}_2) = 2$, however, $d_\text{NN}(\vtr{x}_2, \vtr{x}_1) = 4$, violating symmetry. Furthermore, $d_\text{NN}(\vtr{x}_2, \vtr{w}) = 1$ and $d_\text{NN}(\vtr{w}, \vtr{x}_1) = 2$, implying that $d_\text{NN}(\vtr{x}_2, \vtr{x}_1) > d_\text{NN}(\vtr{x}_2, \vtr{w}) + d_\text{NN}(\vtr{w}, \vtr{x}_1)$, violating the triangle inequality.\footnote{The triangle inequality seems less of an issue, as we might be okay with $d_\text{NN}(\vtr{x}_2, \vtr{x}_1)$ being bounded by a function of $d_\text{NN}(\vtr{x}_2, \vtr{w})$ and $d_\text{NN}(\vtr{w}, \vtr{x}_1)$. See~\cite{chatzikokolakis2013broadening, fernandes2019generalised}.} Thus, the resulting mechanism cannot be $\dx$-private. 

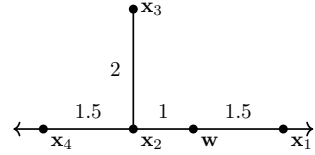
\begin{figure}[ht]
\centering
\resizebox{.5\linewidth}{!}{
\begin{tikzpicture}
\draw[thick,<->] (-3,0) -- (2,0);
\draw[thick,-] (-1,0) -- (-1,2);
\tkzDefPoint(0,0){w};
\tkzLabelPoint[below right](w){$\vtr{w}$};
\node at (w)[circle,fill,inner sep=1.5pt]{};

\tkzDefPoint(1.5,0){x1};
\tkzLabelPoint[below right](x1){$\vtr{x}_1$};
\node at (x1)[circle, fill, inner sep=1.5pt]{};

\tkzDefPoint(0.75,0.3){d1};
\node at (d1)[]{$1.5$};

\tkzDefPoint(-1,2){x3};
\tkzLabelPoint[right](x3){$\vtr{x}_3$};
\node at (x3)[circle, fill, inner sep=1.5pt]{};

\tkzDefPoint(-1.3,1){d3};
\node at (d3)[]{$2$};

\tkzDefPoint(-1,0){x2};
\tkzLabelPoint[below right](x2){$\vtr{x}_2$};
\node at (x2)[circle, fill, inner sep=1.5pt]{};

\tkzDefPoint(-0.5,0.3){d2};
\node at (d2)[]{$1$};

\tkzDefPoint(-2.5,0){x4};
\tkzLabelPoint[below right](x4){$\vtr{x}_4$};
\node at (x4)[circle, fill, inner sep=1.5pt]{};

\tkzDefPoint(-1.75,0.3){d4};
\node at (d4)[]{$1.5$};

\end{tikzpicture}
}
\caption{Example used to illustrate the fact that $d_\text{NN}$ is not a metric. The numbers show the distance of the corresponding segments.}
\label{fig:nn-not-metric}
\end{figure}

Theorem~\ref{theorem:loss-function} says that on average the nearest neighbor of the noisy embedding $\vtr{w}^*$ is the original word $\vtr{w}$, followed by $\vtr{w}$’s nearest neighbor, then its second nearest neighbor, and so on. Due to large probability differences, as shown in Figure~\ref{fig:vocab-cdf}, the original word is chosen overwhelmingly. To mitigate this, our idea is to \emph{not} select the nearest neighbor of $\vtr{w}^*$ every time, and instead occasionally sample other neighbors. This could be used as a post-processing step \emph{after} we have found the nearest word $\vtr{x}^*$ to the noisy embedding $\vtr{w}^*$ through Eq.~\eqref{eq:min-problem}. In other words, we sort the nearest neighbors of $\vtr{x}^*$ and output a neighbor proportional to $\exp(-\epsilon d_\text{NN}(\vtr{x}^*, \vtr{x}))$. More specifically any word $\vtr{x} \in \mathcal{D}$ is output with probability:
\[
\frac{\exp(- c \epsilon d_\text{NN}(\vtr{x}^*, \vtr{x}))}{\sum_{\vtr{x} \in \mathcal{D}} \exp(- c \epsilon d_\text{NN}(\vtr{x}^*, \vtr{x}))}, 
\]
where $c$ is a constant to control how many neighbors are likely to be selected. Note that $\epsilon$ is not used here for privacy protection but rather to ensure that the mechanism behaves as expected, e.g., only the original word output with very high values of $\epsilon$. A higher value such as $c > 1$ means that the mechanism will output the first few neighbors with high probability, and a lower value such as $c = 0.01$ means that more neighbors will likely to be output, of course, with probability exponentially decreasing as we move away from the original word. This is the same as the temperature variable in the softmax function. Note that we cannot use the original word $\vtr{w}$ in the above expression in place of $\vtr{x}^*$, as the resulting fix will no longer be a post-processing step (it uses the knowledge of the original word $\vtr{w}$).

Figure~\ref{fig:fix} shows the result of applying this fix to the \verb+GloVe-Wiki+ and \verb+Word2Vec+ vocabularies. The value of $c$ is chosen by trying different values and choosing the one that gives the best result in terms of the proportion of the times the original word, its 100 nearest neighbors and distant neighbors are output by the mechanism. This value can be pre-computed for each vocabulary. Compare this to Figure~\ref{fig:lap-md}, where between $\epsilon= 10$ and $20$ for \verb+GloVe-Wiki+ and between $\epsilon = 50$ and $60$ for \verb+Word2Vec+, either the original word or its distant neighbors are output by the original mechanism. Our mechanism in comparison ensures a more equitable proportion for the original word, its close and distant neighbors, in line with the one-dimensional case shown in Figure~\ref{fig:lap-1d}. A drawback of this mechanism, of course, is that the value of $c$ needs to be empirically determined for each vocabulary.

\begin{figure}[]
    \centering
    \includegraphics[width=\columnwidth]{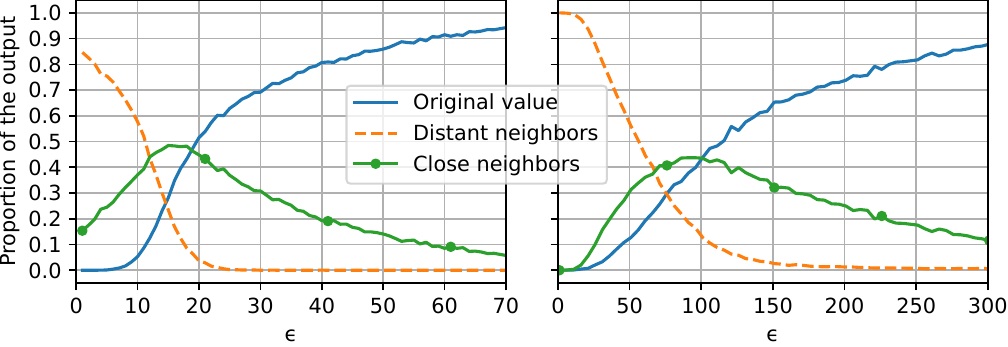}
    \caption{The proportion of times the original word, its close neighbors and distant neighbors are output by the fixed mechanism on the \texttt{GloVe-Wiki} (300 dimensions) with $c=0.04$ (left) and \texttt{Word2Vec} with $c=0.007$ (right). Close neighbors are the first 100 nearest neighbors, and all other words are distant neighbors.}
    \label{fig:fix}
\end{figure}

\begin{figure}[]
    \centering
    \includegraphics[width=\columnwidth]{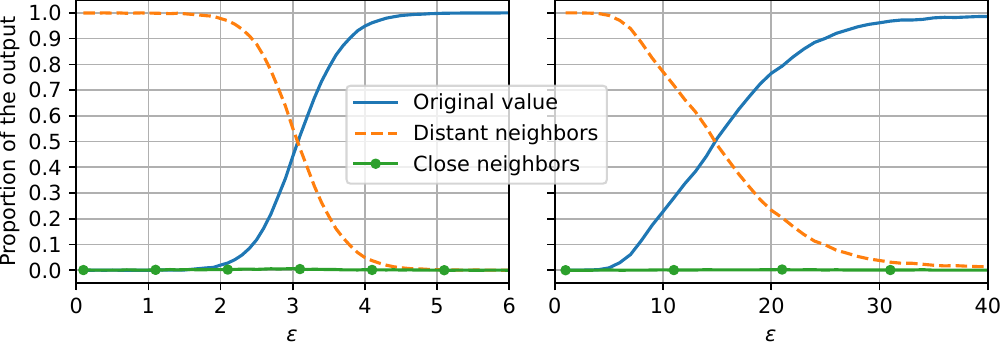}
    \caption{The proportion of times the original word, its close neighbors and distant neighbors are output by the exponential mechanism of \cite{yue2021dxprivacy} on the \texttt{GloVe-Wiki} (300 dimensions) (left) and \texttt{Word2Vec} (right). Close neighbors are the first 100 nearest neighbors, and all other words are distant neighbors.}
    \label{fig:exponentialYueEtAl}
\end{figure}

\subsubsection*{Comparison with the Exponential Mechanism}
Our fix looks similar to the exponential mechanism proposed by Yue et al.~\cite[Algorithm 1]{yue2021dxprivacy}. They pre-compute a matrix containing the probability that each token of the vocabulary is replaced by any other token. The probabilities are computed as an exponential of the Euclidean distance. More precisely, the probability that a token $\vtr{w}$ is replaced by a token $\vtr{x}$ is proportional to $\exp(-\frac{1}{2}\epsilon\cdot \norm{\vtr{w} - \vtr{x}}))$.
The authors argue that the mechanism has better performance by avoiding the nearest neighbor search in an $n$-dimensional space during sanitization. However, in Figure~\ref{fig:exponentialYueEtAl} we show that it still suffers from the problem exposed by us, namely that close neighbors are almost never sampled. The main difference between this mechanism and our fix is that we use the rank instead of a distance metric to compute probabilities. We leave it as an open problem to construct a distance metric that ``flattens out'' the distance between an embedding and its nearest neighbor, and distances between its consecutive neighbors. That is, it should behave similarly to the nearest neighbor function $d_\text{NN}$, while still satisfying the properties of a distance metric. 

\section{Related Work}
\label{sec:rw}
In addition to the multidimensional Laplace mechanism for word-level $\dx$-privacy from~\cite{fernandes2019generalised, feyisetan2020dxprivacy, qu2021bert, yue2021dxprivacy, li2023llm-prompt} detailed in this paper, a few other word-level $\dx$-privacy methods have also been proposed in the literature. In~\cite{feyisetan2019hyperbolic}, the authors propose a $\dx$-private mechanism using a distance metric in the hyperbolic space, which according to the authors, better preserves hierarchical relationships between words. For instance, the hierarchical relationship of the city of London to England. Xu et al~\cite{xu2020mahalanobis} propose the Mahalanobis mechanism, which instead of the spherical noise via the multidimensional Laplace mechanism, samples elliptical noise. The authors note that this mechanism provides better privacy for isolated points (see Section~\ref{subsec:illustration}) since the elliptical nature of the noise results in sampling points other than the original word with higher probability. Another mechanism for word-level $\dx$-privacy is the truncated exponential mechanism (TEM)~\cite{carvalho2023tem}, proposed to remove the drawback of the multidimensional Laplace mechanism which adds the same amount of noise regardless of whether a word is in a dense or sparse region of space. Applying their mechanism to a data domain would add more noise to points in low density areas. Although the mechanisms from~\cite{carvalho2023tem} and \cite{xu2020mahalanobis} may solve the issue with isolated embeddings, they are unlikely to resolve the issue raised by us. As we show in Section~\ref{sec:loss-function}, the problem stems from the large difference in probabilities of sampling the same word versus its nearest neighbors in the nearest neighbor search. The Mahalanobis mechanism~\cite{xu2020mahalanobis} also contains  nearest neighbor search as a post-processing step, and the TEM from~\cite{carvalho2023tem} uses a distance metric to sample a word. Thus the problem is likely to persist in these mechanisms. 

Instead of word-level differential privacy, several works focus on providing differential privacy at the sentence level~\cite{mattern-limits-word-level-dp, igamberdiev2023dp-bart}. More specifically, these works calculate the conditional probabilities of the next word, given a sequence of previous words. We note that both the mechanisms in~\cite{mattern-limits-word-level-dp} and \cite{ igamberdiev2023dp-bart} are for ordinary differential privacy and not $\dx$-privacy. Finally, differential privacy mechanisms for more advanced downstream tasks such as paraphrasing have also been proposed~\cite{utpala2023document-level}, but again using ordinary differential privacy. The exponential mechanism in~\cite{mattern-limits-word-level-dp} and~\cite{utpala2023document-level} is similar to our fix for multidimensional Laplace mechanism for $\dx$-privacy with the difference being that our fix is a post-processing step.

There have been some prior works showing the weaknesses of word-level $\dx$-privacy. The aforementioned work in~\cite{mattern-limits-word-level-dp} criticizes several aspects of word-level $\dx$-privacy, including introduction of grammatical errors and only making changes to individual words, rather than changing lengths of sentences. The work in \cite{pang2024reconstruction} demonstrates attacks on word and sentence-level differential privacy on text data, by reconstructing santized prompts where the original prompts are taken from the training data of the language model. The issues and attacks mentioned in these two works are tangential to the problem addressed in our paper. 

The fact that $\dx$-privacy, in particular, the multidimensional Laplace mechanism does not provide adequate protection for isolated points in the input domain has also been highlighted in the case of location privacy~\cite{chatzikokolakis2015elastic, biswas2023privic}. We discuss here why we think these solutions are not readily applicable to fix our issue in a high-dimensional space. The works in~\cite{chatzikokolakis2015elastic, biswas2023privic} use the concept of \emph{elastic distinguishability} introduced in~\cite{chatzikokolakis2015elastic}. The idea is that in addition to the usual requirement in $\dx$-privacy that nearby points should have higher probability of being sampled than far away points, the mechanism should also sample a point proportional to its \emph{probability mass} or density. For instance, an isolated location such as an island has a lower probability mass than a location in a dense urban area~\cite{chatzikokolakis2015elastic}. The mechanisms in~\cite{chatzikokolakis2015elastic, biswas2023privic} are essentially weighted versions of the exponential mechanism, where the probability of sampling a point is also weighted by its probability mass. Location has a natural candidate for assigning a probability mass to points in space: the number of people at a given location. 

First, we emphasize that their observation is related to the (location) data distribution, with the assumption that the bulk of data is concentrated with only a few isolated points. This is unlike our finding where we show that in essence every embedding in high dimensions is isolated. Second, in the context of location data, even isolated points have a non-zero probability of being output by the mechanism, as they are valid locations. Consequently, there is no need for a nearest neighbor search. Indeed the mechanisms in ~\cite{chatzikokolakis2015elastic, biswas2023privic} do not use it. This is unlike the word embedding space, where a random vector is overwhelmingly not an actual word embedding, and hence the need to perform the nearest neighbor search which, as we saw, is what causes the issue identified in this paper. Third, if we were to apply a similar mechanism to our setting, we need to determine how to assign probability mass to regions in the embedding space. One way to do this is to divide the embedding space into equal-volume hypercubes, count the number of embeddings falling in each hypercube and assign the fractional count as the probability mass. This is analogous to how the two-dimensional location map is divided into a grid in~\cite{chatzikokolakis2015elastic, biswas2023privic}. However, dividing the embedding space in hypercubes is not trivial. On the one hand, dividing each of the $n$ dimensions into $m$ shares results in $m^n$ hypercubes which is unfit for storage (e.g., $n=300$ in some of our vocabularies). On the other hand, choosing a subset of the $n$ dimensions to be divided requires an analysis of which dimension to divide or not.
Thus it is unclear whether the notion of elastic distinguishability offers a solution let alone a feasible one.

The work in~\cite{athanasiou2025enhancingmetric} also looks at the distinguishability of isolated points in the \emph{shuffle model} of differential privacy. In the shuffle model, an intermediate shuffler permutes (local) differentially private inputs of end users before submitting them to the central aggregator. This amplifies privacy compared to the local model as the server cannot tie an output to a particular user. In the $\dx$-privacy equivalent version of the shuffle model, the server can still identify isolated inputs (even when perturbed) as they stand out from the rest. To remove this drawback, one of the techniques proposed in~\cite{athanasiou2025enhancingmetric} is to encode an input in unary and then apply differentially private noise (randomized response) to each bit before sending it to the shuffler. It is not clear how these techniques can be applied to the case of text embeddings. We would need to convert embeddings into bit vectors and then combine multiple embeddings together. Moreover, once again, the isolated point issue is endemic in high dimensional embeddings as opposed to being just a few bad apples.  

Lastly, researchers have identified issues with Euclidean distance in high dimensions. Beyer et al~\cite{beyer1999nn-useful} show that the distance of the nearest neighbor approaches that of the farthest neighbor as the number of dimensions increases, thus making nearest neighbor search using the Euclidean distance meaningless, a result known as \emph{concentration of distances}~\cite{durrant2009nn-useful-converse}. The authors in~\cite{durrant2009nn-useful-converse}, expand on this to show that the concentration of distances is not applicable as long as the number of `relevant' dimensions are on par with the actual data dimensions, and hence in these cases nearest neighbor search is still meaningful. This seems to be the case with word embedding models as is backed up by the empirical results in this paper and the fact that word embedding models are trained to ensure that Euclidean distance can be used to find similar words even with high dimensions~\cite{glove, word2vec}. For more information on this topic, we refer the reader to the survey~\cite{zimek2012survey-hd}.

\section{Conclusion}
The multidimensional Laplace mechanism for $\dx$-privacy is widely used as a method to provide privacy for sensitive text data due to its ease of implementation. We have shown that the method behaves unexpectedly under common word embedding models. More specifically, it almost never outputs semantically related words, as it is expected to do for utility. We have extensively analyzed the noise generated through this mechanism and ruled out any issues with the original mechanism. Instead, we have identified the post-processing step of the nearest neighbor search as the culprit, causing each word embedding to behave as an outlier in high dimensions. We have provided an easy-to-use fix which makes the mechanism behave more expectedly. 
%Our code is available online\hyperref[footnoteCodeRepo]{\textsuperscript{\ref{footnoteCodeRepo}}} and 
The reader is invited to investigate alternative remedies.   

\begin{acks}
This research received no specific grant from any funding agency in the public, commercial, or not-for-profit sectors.
\end{acks}
% \begin{acks}
% \end{acks}

\bibliographystyle{ACM-Reference-Format}
\bibliography{dx-pitfalls}

\appendix

\section{Proofs}
\label{app:proof}
\subsubsection*{Proof of Theorem~\ref{theorem:dist-flip}}
\begin{proof}
    \begin{align}
       & \norm{\vtr{w}^* - \vtr{y}}^2 > \norm{\vtr{w}^* - \vtr{x}}^2 \nonumber\\
&\Rightarrow        \norm{\vtr{w} - \vtr{y} + \noise}^2 > \norm{\vtr{w} - \vtr{x} + \noise}^2 \nonumber\\
&\Rightarrow        \norm{\vtr{w} - \vtr{y}}^2 + \norm{\noise}^2 + 2 \dotprod{\vtr{w} - \vtr{y}}{\noise} \nonumber\\
&> \norm{\vtr{w} - \vtr{x}}^2 + \norm{\noise}^2 + 2 \dotprod{\vtr{w} - \vtr{x}}{\noise} \nonumber\\
&\Rightarrow      \frac{1}{2} ( \norm{\vtr{w} - \vtr{y}}^2 - \norm{\vtr{w} - \vtr{x}}^2 )  >  \dotprod{\vtr{w} - \vtr{x}}{\noise} - \dotprod{\vtr{w} - \vtr{y}}{\noise} \nonumber\\
&= \dotprod{\vtr{w}}{\noise} -  \dotprod{\vtr{x}}{\noise} - \dotprod{\vtr{w}}{\noise} + \dotprod{\vtr{y}}{\noise} \nonumber\\
&= \dotprod{\vtr{y} - \vtr{x}}{\noise} \nonumber\\
&= r \norm{\vtr{y} - \vtr{x}} \cos \theta_{\vtr{\eta, \vtr{y} - \vtr{x}}} \label{eq:cond-2}\\
&\Rightarrow \frac{1}{2} ( \norm{\vtr{w} - \vtr{y} +\vtr{x} - \vtr{x}}^2 - \norm{\vtr{w} - \vtr{x}}^2 ) > r \norm{\vtr{y} - \vtr{x}} \cos \theta_{\vtr{\eta, \vtr{y} - \vtr{x}}} \nonumber\\
&\Rightarrow \frac{1}{2} ( \norm{\vtr{w} - \vtr{x}}^2 + \norm{\vtr{y} - \vtr{x}}^2 - 2\dotprod{\vtr{w} - \vtr{x}}{\vtr{y} - \vtr{x})} - \norm{\vtr{w} - \vtr{x}}^2 ) \nonumber\\
&> r \norm{\vtr{y} - \vtr{x}} \cos \theta_{\vtr{\eta, \vtr{y} - \vtr{x}}} \nonumber\\
&\Rightarrow \frac{1}{2} ( \norm{\vtr{y} - \vtr{x}}^2 - 2\dotprod{\vtr{w} - \vtr{x}}{\vtr{y} - \vtr{x})} ) > r \norm{\vtr{y} - \vtr{x}} \cos \theta_{\vtr{\eta, \vtr{y} - \vtr{x}}} \nonumber\\
&\Rightarrow \frac{1}{2} ( \norm{\vtr{y} - \vtr{x}}^2 - 2 \norm{\vtr{w} - \vtr{x}}\norm{\vtr{y} - \vtr{x})} \cos \theta_{\vtr{w} - \vtr{x}, \vtr{y} - \vtr{x}}) \nonumber \\
&> r \norm{\vtr{y} - \vtr{x}} \cos \theta_{\vtr{\eta, \vtr{y} - \vtr{x}}} \nonumber\\
&\Rightarrow \frac{1}{2} \norm{\vtr{y} - \vtr{x}} > \norm{\vtr{w} - \vtr{x}} \cos \theta_{\vtr{w} - \vtr{x}, \vtr{y} - \vtr{x}} +  r \cos \theta_{\vtr{\noise, \vtr{y} - \vtr{x}}}\nonumber
        \end{align}
as required. 
%Rearranging the above we obtain
% \[
% \frac{1}{2} \norm{\vtr{y} - \vtr{x}}  <  \norm{\vtr{w} - \vtr{x}} \cos \theta_{\vtr{w} - \vtr{x}, \vtr{y} - \vtr{x}} + r \cos \theta_{\vtr{\noise, \vtr{y} - \vtr{x}}} 
% \]
\end{proof}

\subsubsection*{Proof of Theorem~\ref{theorem:beta-moments}}
\begin{proof}
Let $B_n = B((n-1)/2, 1/2)$. For any $n \geq 4$, we note that:
\begin{align}
    \frac{B_{n-2}}{B_n} &= \frac{\Gamma(\frac{n-3}{2})\Gamma(\frac{1}{2})}{\Gamma(\frac{n-3}{2} + \frac{1}{2})} \frac{\Gamma(\frac{n-1}{2} + \frac{1}{2})}{\Gamma(\frac{n-1}{2})\Gamma(\frac{1}{2})} \nonumber\\
    &= \frac{\Gamma(\frac{n-3}{2})}{\Gamma(\frac{n}{2} - 1)} \frac{\Gamma(\frac{n}{2} )}{\Gamma(\frac{n-1}{2})} \nonumber\\
    &= \frac{\Gamma(\frac{n-3}{2})}{\Gamma(\frac{n}{2} - 1)} \frac{\Gamma(\frac{n}{2} -1  + 1)}{\Gamma(\frac{n-1}{2} -1 + 1)}\nonumber\\
    &= \frac{\Gamma(\frac{n-3}{2})}{\Gamma(\frac{n}{2} - 1)} \frac{(\frac{n}{2} -1) \Gamma(\frac{n}{2} -1)}{(\frac{n-3}{2})\Gamma(\frac{n-3}{2})}\nonumber\\
    & = \frac{n-2}{n-3} \label{eq:beta-rel},
\end{align}
    where we have used the fact that $\Gamma(x + 1) = x \Gamma(x)$ for any real number $x > 0$. We first consider the cases of $n = 2$ and $3$. When $n = 2$, we have:
\begin{align*}
    \mu(j, 2) &= \int_{-1}^{+1} k^j \frac{1}{B_2} (1 - k^2)^{\frac{0}{2} - 1} dk\\
    &=\frac{1}{\pi} \int_{-1}^{+1} \frac{k^j}{\sqrt{1 - k^2}} dk,
\end{align*}
where we have used the fact that $B_2 = B(1/2, 1/2) = \Gamma(1/2) \Gamma(1/2)/\Gamma(1) = \sqrt{\pi} \cdot \sqrt{\pi}/(1) = \pi$. Let $k = \sin t$. Then $dk = \cos t \; dt$. Therefore,
\begin{align}
    \mu(j, 2) &= \frac{1}{\pi} \int^{\pi/2}_{-\pi/2} \frac{\sin^j t}{\sqrt{1 - \sin^2 t}} \cos t \; dt \nonumber \\
    &=  \frac{1}{\pi} \int^{\pi/2}_{-\pi/2} \frac{\sin^j t}{| \cos t|}  \cos t \; dt \nonumber\\
    &= \frac{1}{\pi} \int^{\pi/2}_{-\pi/2} {\sin^j t} \; dt \nonumber,
\end{align}
where we have used the fact that $\cos t$ is positive in the interval $(-\pi/2, \pi/2)$. With $j = 0$, it is easy to see that $\mu(0, 2) = 1$. And for $j = 1$, we see that $\mu(1, 2) = 0$ since the integral of $\sin t$ is $-\cos t$. Consider therefore $j \geq 2$. Integrating by parts, we have
\begin{align}
    \mu(j, 2) &=  \frac{1}{\pi} \int^{\pi/2}_{-\pi/2}  \sin{t} \;  {\sin^{j-1} t} \; dt, \nonumber\\
    &=  - \frac{1}{\pi} \cos t \;  {\sin^{j-1} t}  \Big|^{\pi/2}_{-\pi/2} \nonumber\\
    &- \frac{1}{\pi} \int^{\pi/2}_{-\pi/2} -\cos t \;  (j - 1) \; {\sin^{j-2} t} \;  \cos t \; dt , \nonumber\\
    &= \frac{j-1}{\pi} \int^{\pi/2}_{-\pi/2} \cos^2 t \;   {\sin^{j-2} t} \; dt \nonumber\\
    &= \frac{j-1}{\pi} \int^{\pi/2}_{-\pi/2} (1 - \sin^2 t) \;   {\sin^{j-2} t} \; dt \nonumber\\
    &= \frac{j-1}{\pi} \int^{\pi/2}_{-\pi/2} {\sin^{j-2} t}  \; dt - \frac{j-1}{\pi} \int^{\pi/2}_{-\pi/2} {\sin^{j} t} \; dt \nonumber\\
    &= (j-1)\mu(j-2, 2) - (j-1) \mu(j, 2) \nonumber \\
\Rightarrow \mu(j, 2) &= \frac{j-1}{j} \mu(j-2, 2) \label{eq:mu-j-recurrence}
\end{align}
From this recurrence relation, we see that for odd $j$ we have $\mu(j, 2) = 0$, since $\mu(1, 2) = 0$. And for even $j$, by using $\mu(0, 2) = 1$ in the recurrence relation, it is easy to show that
\[
\mu(j, 2) = \frac{(j-1)!!}{j !!},
\]
where $x !! = x(x-2)(x-4) \cdots 4 \cdot 2$ is the double factorial. Thus,
\begin{equation}
\label{eq:mu-j-1-recurrence}
    \mu(j, 2) = \begin{cases}
        1, \; \text{ if } j = 0,\\
        0, \; \text{ if } j \text{ is odd},\\
        \frac{(j-1)!!}{j !!}, \; \text{ if } j \text{ is even}
    \end{cases}
\end{equation}
Next consider $n = 3$. We have
\begin{align}
  \mu(j, 3) &= \int_{-1}^{+1} k^j \frac{1}{B_3} (1 - k^2)^{\frac{2}{2} - 1} \; dk \nonumber\\
    &=\frac{1}{2} \int_{-1}^{+1} {k^j} \; dk, \nonumber\\
     &=\frac{1}{2}  \frac{k^{j+1}}{j+1} \Big|^{+1}_{-1} \nonumber\\
     &= \frac{1 - (-1)^{j + 1}}{2(j + 1)}, \nonumber
\end{align}
where we have used the fact that $B_3 = B(2/2, 1/2) = B(1, 1/2) = \Gamma(1) \Gamma(1/2)/\Gamma(1 + 1/2) = \Gamma(1/2)/((1/2) (\Gamma(1/2)) = 2$. Thus,
\begin{equation}
\label{eq:mu-j-2-recurrence}
    \mu(j, 3) = \begin{cases}
        1, \; \text{ if } j = 0,\\
        0, \; \text{ if } j \text{ is odd},\\
        \frac{1}{j+1}, \; \text{ if } j \text{ is even}
    \end{cases}
\end{equation}
Finally, consider $n \geq 4$. We have through integration by parts:
\begin{align}
  \mu(j, n) &= \int_{-1}^{+1} k^j \frac{1}{B_n} (1 - k^2)^{\frac{n-1}{2} - 1} \; dk \nonumber\\
  &= \frac{1}{B_n} \int_{-1}^{+1} k^j  (1 - k^2)^{\frac{n-1}{2} - 1} \; dk \nonumber\\
  &= \frac{1}{B_n}  \frac{k^{j+1}}{j+1}  (1 - k^2)^{\frac{n-1}{2} - 1} \Big|^{+1}_{-1} \nonumber\\
  &- \frac{1}{B_n} \int_{-1}^{+1} \frac{k^{j+1}}{j+1} \left(\frac{n-3}{2}\right)  (1 - k^2)^{\frac{n-1}{2} - 2} (-2k) \; dk \nonumber\\
  &= \frac{n-3}{B_n}\frac{1}{j+1} \int_{-1}^{+1} {k^{j+2}}  (1 - k^2)^{\frac{n-1}{2} - 2}  \; dk \nonumber\\
  &= \frac{n-3}{B_n}\frac{1}{j+1} \int_{-1}^{+1} {k^{j}} (1 - (1 - k^2))  (1 - k^2)^{\frac{n-1}{2} - 2}  \; dk \nonumber\\
  &= \frac{n-3}{B_n}\frac{1}{j+1} \int_{-1}^{+1} {k^{j}} (1 - k^2)^{\frac{n-1}{2} - 2}  \; dk  \nonumber\\
  &- \frac{n-2}{B_n}\frac{1}{j+1} \int_{-1}^{+1} {k^{j}} (1 - k^2)^{\frac{n-1}{2} - 1}  \; dk \nonumber\\
  &= \frac{n-2}{B_{n-2}}\frac{1}{j+1} \int_{-1}^{+1} {k^{j}} (1 - k^2)^{\frac{n-1}{2} - 2}  \; dk \nonumber\\
  &- \frac{n-3}{B_n}\frac{1}{j+1} \int_{-1}^{+1} {k^{j}} (1 - k^2)^{\frac{n-1}{2} - 1}  \; dk \nonumber\\
  &= \frac{(n-2)}{(j+1)}\mu(j, n-2) - \frac{(n-3)}{(j+1)} \mu(j, n) \nonumber
\end{align}
where we have used Eq.~\eqref{eq:beta-rel}. Thus,
\begin{equation}
\label{eq:mu-j-n-recurrence}
     \mu(j, n) = \frac{n-2}{n+j-2}\mu(j, n-2)
\end{equation}
If $j = 0$, then $\mu(0, n)$ is the integral of the PDF of the distribution, and hence $\mu(0, n) = 1$, for all $n \geq 2$. Consider now, odd $j$. From Eqs.~\eqref{eq:mu-j-1-recurrence}, \eqref{eq:mu-j-2-recurrence} and \eqref{eq:mu-j-n-recurrence} we see that $\mu(j, n) = 0$, for all $n \geq 2$. Thus, consider even $j$. First consider even $n$, starting from $n = 4$. We have $\mu(j, n - 2) = \mu(j, 2)$, for $n = 4$. Putting the result from Eq.~\eqref{eq:mu-j-1-recurrence} into Eq.~\eqref{eq:mu-j-n-recurrence}, we get
\begin{align}
    \mu(j, n) &= \frac{n-2}{n+j-2}  \frac{n-4}{n+j-4}  \frac{n-6}{n+j-6} \cdots \frac{4-2}{4+j-2} \frac{(j-1)!!}{j!!} \nonumber\\
    &=\frac{n-2}{n+j-2}  \frac{n-4}{n+j-4}  \frac{n-6}{n+j-6} \cdots \frac{2}{j+2} \frac{(j-1)!!}{j!!} \nonumber\\
    &= \frac{(n-2)!! (j-1)!!}{(n-2+j)!!}. \nonumber
\end{align}
Lastly, consider odd $n$ starting from $n = 5$. We have $\mu(j, n - 2) = \mu(j, 3)$, for $n = 5$. Putting the result from Eq.~\eqref{eq:mu-j-2-recurrence} into Eq.~\eqref{eq:mu-j-n-recurrence}, we get
\begin{align}
    \mu(j, n) &= \frac{n-2}{n+j-2}  \frac{n-4}{n+j-4}  \frac{n-6}{n+j-6} \cdots \frac{5-2}{5+j-2} \frac{1}{j+1} \nonumber\\
    &=\frac{n-2}{n+j-2}  \frac{n-4}{n+j-4}  \frac{n-6}{n+j-6} \cdots \frac{3}{j+3} \frac{1}{j+1} \nonumber\\
    &=\frac{n-2}{n+j-2}  \frac{n-4}{n+j-4}  \frac{n-6}{n+j-6} \cdots \frac{3}{j+3} \frac{1}{j+1} \frac{(j-1)!!}{(j-1)!!} \nonumber \\
    &= \frac{(n-2)!! (j-1)!!}{(n-2+j)!!} \nonumber
\end{align}
Putting the results together, we get for $n \geq 2$
\begin{equation*}
%\label{eq:moments-beta}
    \mu(j, n) = \begin{cases}
        1, \; &\text{ if } j = 0,\\
        0, \; &\text{ if } j \text{ is odd},\\
        \frac{(n-2)!! (j-1)!!}{(n-2+j)!!}, \; &\text{ if } j \text{ is even}
    \end{cases}
\end{equation*}
% From this, using the fact that $K_n = J_{n-1}$ and overloading notation, we get the moments of $K_n$ for all $n \geq 2$ are
% \begin{equation}
% \label{eq:moments-beta}
%     \mu(j, n) = \begin{cases}
%         1, \; &\text{ if } j = 0,\\
%         0, \; &\text{ if } j \text{ is odd},\\
%         \frac{(n-2)!! (j-1)!!}{(n-2+j)!!}, \; &\text{ if } j \text{ is even}
%     \end{cases}
% \end{equation}
From the above equation we see that $\mathbb{E}[K] = \mu(1, n) = 0$ and $\text{Var}[K] = \mathbb{E}[K^2] - (\mathbb{E}[K] )^2 = \mathbb{E}[K^2] = \mu(2, n) = \frac{1}{n} $. 
\end{proof}

\subsubsection*{Proof of Theorem~\ref{theorem:beta-sg}}
\begin{proof}
First note that $\mu = \mu(1, n) = 0$ from Eq.~\eqref{eq:moments-beta}. Therefore, using the Taylor series expansion of the exponential function, Lemma~\ref{lemma:lin-exp} and Eq.~\eqref{eq:moments-beta}, we have:
\begin{align*}
    \mathbb{E}[e^{\lambda (K - \mu)}] &=  \mathbb{E}[e^{\lambda K }]\\
    &= \mathbb{E}\left[
    \sum_{j = 0}^\infty \frac{(\lambda K)^j}{j !} \right] \\
    &= \sum_{j = 0}^\infty  \frac{\lambda^j}{j!} \mathbb{E}[K^j] \\
    &= \sum_{j = 0, j \text{ even}}^\infty  \frac{\lambda^j}{j!} \mathbb{E}[K^j] \\
    &= \sum_{j = 0}^\infty  \frac{(\lambda^2)^j}{(2j)!} \mathbb{E}[K^{2j}] \\
    &= 1 + \sum_{j = 1}^\infty  \frac{(\lambda^2)^j}{(2j)!} \mathbb{E}[K^{2j}] \\
    &= 1 + \sum_{j = 1}^\infty  \frac{(\lambda^2)^j}{(2j)!} \frac{(n-2)!! (2j-1)!!}{(n-2+2j)!!}
\end{align*}
Now, note that 
\begin{align*}
    \frac{(2j - 1)!!}{2j!} &= \frac{(2j - 1)(2j - 3) (2j - 5) \cdots 3 \cdot 1}{2j(2j-1)(2j-2)(2j-3) \cdots 3 \cdot 2 \cdot 1} \\
    &=\frac{1}{2j!!} \\
    &= \frac{1}{2j(2j-2)(2j-4) \cdots 4 \cdot 2}\\
    &= \frac{1}{2^j \cdot j(j-1)(j-2) \cdots 2 \cdot 1}\\
    &= \frac{1}{2^j j!}
\end{align*}
Therefore, the above becomes,
\begin{align*}
    &\mathbb{E}[e^{\lambda (K - \mu)}] \\
    &= 1 + \sum_{j = 1}^\infty  \frac{(\lambda^2)^j}{2^j j!} \frac{(n-2)!!}{(n-2+2j)!!}\\
    &= 1 + \sum_{j = 1}^\infty  \frac{(\lambda^2)^j}{2^j j!} \frac{1}{(n+2j - 2)(n+2j -4) \cdots (n+2) (n)}\\
    &\leq 1 + \sum_{j = 1}^\infty  \frac{(\lambda^2)^j}{2^j j!} \frac{1}{n^j}\\
    &= 1 + \sum_{j = 1}^\infty \frac{1}{j!} \left(\frac{\lambda^2}{2n}\right)^j \\
    &= \sum_{j = 0}^\infty \frac{1}{j!} \left(\frac{\lambda^2}{2n}\right)^j \\
    &= e^{-\sigma^2 \lambda^2/2},
\end{align*}
where $\sigma^2 = \frac{1}{n} = \text{Var}[K]$.
\end{proof}

\subsection*{Proof of Theorem~\ref{theorem:gamma-tail-bound}}
\begin{proof}
Recall first that Markov's inequality states that if $X$ is a non-negative random variable and $a > 0$, then 
\[
\Pr[X \geq a ] \leq \frac{\mathbb{E}[X]}{a}
\]
Also, let $M_X(t) = \mathbb{E}[e^{tX}]$ be the moment generating function of $X$, where we assume the expectation to exist within $-h < t < h$, for some real number $h > 0$~\cite[\S 2.3]{casella2024statistical}. 

Let us define another random variable $Y = e^{tX}$. Since $e^{tx}$ is an increasing function of $x$ if $t > 0$, we have that 
\begin{equation}
\label{eq:mgf-t-pos}
    \Pr[X \geq a ] = \Pr[e^{tX} \geq e^{ta} ] \leq  \frac{\mathbb{E}[e^{tX}]}{e^{ta}} = e^{-ta}M_X(t),
\end{equation}
where we have used Markov's inequality on the random variable $Y$ with $0 < t < h$. Next $e^{tx}$ is a decreasing function of $x$ if $t < 0$. Therefore
\begin{equation}
\label{eq:mgf-t-neg}
\Pr[X \leq a ] = \Pr[e^{tX} \geq e^{ta} ] \leq  \frac{\mathbb{E}[e^{tX}]}{e^{ta}} = e^{-ta}M_X(t), 
\end{equation}
with $-h < t < 0$. The moment generating function $M_\text{G}$ of $R$ is~\cite[\S 2.3]{casella2024statistical}
\begin{equation}
\label{eq:mgf-gamma}
M_\text{G}(t) = (1 - t/\epsilon)^{-n},  \quad t < \epsilon.    
\end{equation}
Let $c > 1$, $a = cn/\epsilon$ and $R = X$, then for $0 < t < \epsilon$ through Eqs.~\eqref{eq:mgf-gamma} and \eqref{eq:mgf-t-pos} we get
\begin{align*}
    \Pr\left[R \geq \frac{cn}{\epsilon}\right] &\leq e^{-cnt/\epsilon} \left(1 - \frac{t}{\epsilon}\right)^{-n} = y.
\end{align*}
To find the minimum value of $y$ in the interval $0 < t < \epsilon$, we take the derivative of $y$, which gives
\begin{align*}
    y' &= -\frac{cn}{\epsilon} e^{-cnt/\epsilon} \left(1 - \frac{t}{\epsilon}\right)^{-n} \\
    &+ e^{-cnt/\epsilon} (-n) \left(1 - \frac{t}{\epsilon}\right)^{-n-1} \left( -\frac{1}{\epsilon} \right) \\
    &= -\frac{cn}{\epsilon} e^{-cnt/\epsilon} \left(1 - \frac{t}{\epsilon}\right)^{-n} + \frac{n}{\epsilon} e^{-cnt/n\epsilon}\left(1 - \frac{t}{\epsilon}\right)^{-n-1}  \\
    &= \frac{n}{\epsilon} e^{-cnt/\epsilon} \left(1 - \frac{t}{\epsilon}\right)^{-n-1} \left( -c \left(1 - \frac{t}{\epsilon}\right) + 1\right)  \\
    &= \frac{n}{\epsilon} e^{-cnt/\epsilon} \left(1 - \frac{t}{\epsilon}\right)^{-n-1} \left( \frac{ct}{\epsilon} - c + 1\right). 
\end{align*}
We see that the derivative is $0$ if $t = \frac{(c-1)\epsilon}{c}$, which is in the interval $(0 , \epsilon)$ if $c > 1$. If we take the second derivative of $y$, we get 
\begin{align*}
    y''= &-\frac{cn^2}{\epsilon^2} e^{-cnt/\epsilon} \left(1 - \frac{t}{\epsilon}\right)^{-n-1} \left( \frac{ct}{\epsilon} - c + 1\right) \\ &+ \frac{cn}{\epsilon} e^{-cnt/\epsilon} \left(1 - \frac{t}{\epsilon}\right)^{-n-1} \\
    &+ \frac{n(n+1)}{\epsilon^2} e^{-cnt/\epsilon} \left(1 - \frac{t}{\epsilon}\right)^{-n-2} \left( \frac{ct}{\epsilon} - c + 1\right).
\end{align*}
At $t = \frac{(c-1)\epsilon}{c}$, the first and the last term in the above expression is 0. The remaining term is positive since the exponential function is greater than 0 for any value of $t$, and $1 - t/\epsilon$ is also greater than $0$ if $0 < t < \epsilon$. Thus, $y'' > 0$. Hence the function $y$ is minimized at this value of $t$. Putting this value of $t$ into the expression for $y$, we get the first statement of the theorem. 

For the second statement, again let $c > 1$, $a = n/c\epsilon$ and $R = X$, then for $-h < t < 0$ through Eqs.~\eqref{eq:mgf-gamma} and \eqref{eq:mgf-t-neg} we get
\[
\Pr\left[R \leq \frac{n}{c\epsilon}\right] \leq e^{-nt/c\epsilon} \left(1 - \frac{t}{\epsilon}\right)^{-n} = y.
\]
Again, taking the derivative of $y$ to find the value of $t$ that minimizes $y$ in the intervale $-h < t < 0$, we get
\begin{align*}
    y' = \frac{n}{\epsilon} e^{-nt/c\epsilon} \left(1 - \frac{t}{\epsilon}\right)^{-n-1} \left( \frac{t}{c\epsilon} - \frac{1}{c} + 1\right).
\end{align*}
We see that $y' = 0$ if $t = (1-c)\epsilon$, which indeed satisfies $t < 0 < \epsilon$, with $c > 1$. Taking the second derivative of $y$, we get:
\begin{align*}
    y'' = &-\frac{n^2}{c\epsilon^2} e^{-nt/c\epsilon} \left(1 - \frac{t}{\epsilon}\right)^{-n-1} \left( \frac{t}{c\epsilon} - \frac{1}{c} + 1\right) \\ &+ \frac{n}{c\epsilon} e^{-nt/c\epsilon} \left(1 - \frac{t}{\epsilon}\right)^{-n-1} \\
    &+ \frac{n(n+1)}{\epsilon^2} e^{-nt/c\epsilon} \left(1 - \frac{t}{\epsilon}\right)^{-n-2} \left( \frac{t}{c\epsilon} - \frac{1}{c} + 1\right).
\end{align*}
Once again at $t = (1-c)\epsilon$, the first and the last term in the above expression vanishes. The middle term is positive since the exponential function is positive for any value of $t$ and $1 - t/\epsilon = c > 1$ for this value of $t$. Thus, $y'' > 0$, meaning that $y$ is minimized at this value of $t$. Putting this value of $t$ in the expression for $y$ gives us the second statement of the theorem. 
\end{proof}

\section{Noise Length Follows the Gamma Distribution}
\label{app:sec:gamma}
Let $f(r, S)$ denote the joint probability density function of the noise distribution, which samples a noise vector $\noise = r \hat{\vtr{u}}$ where $\hat{\vtr{u}}$ is sampled uniformly over the surface of the $n$-dimensional hypersphere of unit radius, and where $r = \norm{\noise}$, whose distribution we seek to determine. Let $f_R(r)$ denote the probability density function of this distribution. This is given by the marginal PDF
\begin{equation}
\label{eq:marginal-pdf}
 f_R(r) = \int_{S(r)} f(r, S) \; dS,   
\end{equation}
where the integration is over $S(r)$, i.e., the surface area of the $n$-dimensional hypersphere of radius $r$. Now, $S(r) \propto r^{n-1}$ and we want the distribution at $r$ to be proportional to $\exp(-\epsilon \norm{\noise}) = \exp(-\epsilon r)$. Therefore, $f_R(r) \propto r^{n-1} \exp(-\epsilon r)$. To make this into a probability density function, we must have:
\begin{align*}
    \int_0^{\infty} c r^{n-1} e^{-\epsilon r} \; dr = c \int_0^{\infty}  r^{n-1} e^{-\epsilon r} \; dr = 1
\end{align*}
Let $I_{n-1} = \int_0^{\infty}  r^{n-1} e^{-\epsilon r} \; dr$. Then
\begin{align*}
    I_{n-1} &= \int_0^{\infty}  r^{n-1} e^{-\epsilon r} \; dr\\
    &= \frac{e^{-\epsilon r}}{-\epsilon} r^{n-1} \Big|_0^\infty - \int_0^{\infty} \frac{e^{-\epsilon r}}{-\epsilon} (n-1) r^{n-2} \; dr \\
    &= \frac{n-1}{\epsilon} \int_0^{\infty} e^{-\epsilon r} r^{n-2} \; dr \\
    &= \frac{n-1}{\epsilon} I_{n-2}.
\end{align*}
Now, 
\[
    I_1 = \int_0^{\infty}  e^{-\epsilon r} \; dr =  \frac{e^{-\epsilon r}}{-\epsilon} \Big|_0^\infty = -\frac{1}{\epsilon}(0 - 1) = \frac{1}{\epsilon}.
\]
Therefore,
\begin{align*}
    I_{n-1} &= \frac{(n-1)}{\epsilon}  \frac{(n-2)}{\epsilon} \cdots  \frac{(n-(n-1))}{\epsilon}  \frac{1}{\epsilon} \\
    &= \frac{(n-1)!}{\epsilon^n} = \frac{\Gamma(n)}{\epsilon^n}
\end{align*}
Thus, $c = {\epsilon^n}/\Gamma(n)$, and $f_R(r) = f_\text{G}(r)$, i.e., the gamma distribution given in Eq.~\eqref{eq:gamma-dist}.

As an illustration of this, suppose we want to sample a point uniformly at random on and inside a circle with radius $\rho$. We first sample a point uniformly at random on the circumference of the circle (e.g., via the method of zero mean, unit variance Gaussians as explained in Section~\ref{sec:bg}). To now select a point uniformly at random within the circle (including its circumference), we need to find the length $r$ of the point, where $0 \leq r \leq \rho$. From Eq.~\eqref{eq:marginal-pdf} we see that at radius $r$ the marginal PDF of the random variable $R$ representing the length of the point is given by
\[
f_R(r) = \int_0^{2\pi r} f(r, S) dS,
\]
where $2\pi r$ is the circumference of the circle with radius $r$. See Figure~\ref{fig:circle}.

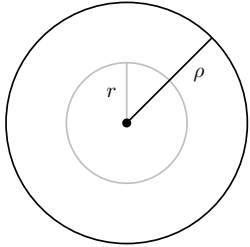
\begin{figure}[ht]
\centering
\resizebox{0.4\linewidth}{!}{
\begin{tikzpicture}

\tkzDefPoint(0,0){o};

\node[circle, thick, draw, minimum size=4cm] (c) at (0,0){};
\tkzDefPoint(1.414, 1.414){u};
%\node at (u)[circle,fill,inner sep=1.5pt]{};

\node[circle, thick, draw, lightgray, minimum size=2cm] (c) at (0,0){};

\draw[thick,-] (o) -- (u);

\tkzDefPoint(0,1){p};
\draw[thick,-, lightgray] (o) -- (p);

\node at (o)[circle,fill,inner sep=1.5pt]{};

\tkzDefPoint(1.2, 0.8){rho};
\node at (rho)[]{$\rho$};

\tkzDefPoint(-0.25,0.5){r};
\node at (r)[]{$r$};

\end{tikzpicture}
}
\caption{Sampling a point uniformly at random on and inside a circle of radius $\rho$. At radius $r$ we need to integrate over the circumference of the inner circle of length $2\pi r$.}
\label{fig:circle}
\end{figure}

Since the point needs to have a uniform distribution, we have that $f_R(r) \propto 2\pi r$. To make it into a probability density function, we get 
\[
\int_0^{\rho} c 2\pi r \; dr = 1 \Rightarrow c = \frac{1}{\pi \rho^2}.
\]
Thus $f_R(r) = \frac{2 \pi r}{\pi \rho^2} = \frac{2r}{\rho^2}$. Indeed, we see that if we let $f_Z(\vtr{z})$ denote the PDF of a point $\vtr{z}$ inside or on the circle of radius $\rho$, i.e., $\norm{\vtr{z}}^2 \leq \rho^2$, and let $f_U(\vtr{u})$ denote the uniform distribution on the circumference of the unit circle which is $1/2\pi$ if $\vtr{u}$ is on the circumference of the unit circle, and 0 otherwise, we get from Eq.~\eqref{eq:prod-ind}:
\begin{align*}
    f_Z(\vtr{z}) &= \int_0^\rho f_R(r) f_U(\vtr{z}/r)\frac{1}{r} \; dr\\
    &= \int_0^\rho \frac{2r}{\rho^2} f_U(\vtr{z}/r)\frac{1}{r} \; dr\\
    &= \frac{2}{\rho^2} \int_0^\rho  f_U(\vtr{z}/r) \; dr\\
    &= \frac{2}{\rho^2} \int_0^\rho \frac{1}{2\pi}  \delta(r - \norm{\vtr{z}}) \; dr\\
     &= \frac{1}{\pi \rho^2} \int_0^\rho \delta(r - \norm{\vtr{z}}) \; dr
    % &= \frac{1}{\pi \rho^2} \int_0^\rho \delta_{\norm{\vtr{z}}}(r) \; dr\\
    % &= \frac{1}{\pi \rho^2} \int_0^\rho \delta_{\norm{\vtr{z}}}(r) \; dr\\   
\end{align*}
where $\delta$ is the Dirac delta function~\cite[\S 9.4]{herman2015delta}. Define $t = r - \norm{\vtr{z}}$, which gives $dt = dr$. Therefore, the above becomes
\begin{align*}
    f_Z(\vtr{z}) = \frac{1}{\pi \rho^2}  \int_{-\norm{\vtr{z}}}^{\rho - \norm{\vtr{z}}} \delta(t) \; dt = \frac{1}{\pi \rho^2} \cdot 1 = \frac{1}{\pi \rho^2},
\end{align*}
where the integral is $1$ because $0 \in [-\norm{\vtr{z}}, \rho - \norm{\vtr{z}}]$~\cite[\S 9.4]{herman2015delta}. This is precisely the area of the circle with radius $\rho$, and hence the distribution is uniform as required. 
\end{document}